\documentclass[journal, onecolumn,10pt,twoside]{IEEEtran}

\usepackage[T1]{fontenc}
\usepackage{amsmath, amsfonts, amssymb, mathrsfs}
\usepackage[amsmath,thmmarks]{ntheorem}
\usepackage{enumitem}
\usepackage{tikz}
\usepackage{booktabs}
\usepackage{bm}
\usepackage{graphicx}
\usepackage{subfigure}
\usepackage{cite}

\usepackage{manfnt}     

\theoremstyle{plain}
\theoremheaderfont{\itshape\bfseries}
\theorembodyfont{\normalfont\itshape}
\theoremseparator{:}
\theoremsymbol{}
\newtheorem{theorem}{Theorem}
\newtheorem{lemma}[theorem]{Lemma}
\newtheorem{corollary}[theorem]{Corollary}
\newtheorem{proposition}[theorem]{Proposition}

\newtheorem{definition}{Definition}

\theoremstyle{plain}
\newtheorem{question}{Question}

\theoremstyle{nonumberplain}
\theoremheaderfont{\itshape}
\theorembodyfont{\normalfont}
\theoremseparator{:}
\theoremsymbol{}
\newtheorem{remark}{Remark}

\theoremstyle{plain}
\theoremheaderfont{\itshape\bfseries}
\theorembodyfont{\normalfont}
\theoremseparator{:}

\def\TT{ \mathbb{T} }							
\def\dD{ \partial\mathbb{D} }			

\def\d{ \mathrm{d} }							
\def\T{ \mathrm{T} }							
\def\FN{ \mathbb{F} }
\def\NN{ \mathbb{N} }													
\def\RN{ \mathbb{R} }													
\def\QN{ \mathbb{Q} }													
\def\CN{ \mathbb{C} }													
\def\RNc{ \mathbb{R}_{\mathrm{c}} }						
\def\B{ \mathcal{B} }								  				
\def\Bc{ \mathcal{B}_{\mathrm{c}} }						
\def\C{ \mathcal{C} }								  				
\def\Cc{ \mathcal{C}_{\mathrm{c}} }						
\def\Csemi{ \mathbf{C}_{1} }									
\def\Cweak{ \mathbf{C}_{2} }									
\def\H12{ H^{1/2} }														
\def\H12c{ H^{1/2}_{\mathrm{c}} }							
\def\W{ \mathcal{A} }													
\def\Energy{\mathrm{E}}
\def\bsx{ \boldsymbol x}											
\def\va{ \bold a}
\def\vx{ \bold x}
\def\vu{ \bold u}
\def\vv{ \bold v}
\def\vzero{ \bold 0}
\def\mNull{ \bold 0}
\def\mA{ \bold A}
\def\mB{ \bold B}
\def\mE{ \bold E}
\def\mF{ \bold F}
\def\mI{ \bold I}
\def\mJ{ \bold J}
\def\mK{ \bold K}
\def\mM{ \bold M}
\def\mN{ \bold N}
\def\mP{ \bold P}
\def\mQ{ \bold Q}

\def\na{ n_{\mathrm{a}} }
\def\nb{ n_{\mathrm{b}} }
\def\Pd{ \Pi_{1} }														
\def\Pu{ \Sigma_{1} }													

\def\A{ \mathcal{A} }													
\def\TM{ \mathrm{TM} }              					

\DeclareMathOperator*{\ind}{ind}			 

\IEEEoverridecommandlockouts

\begin{document}

\title{Singular Linear Differential-Algebraic Equations with Constant Coefficients Necessarily Admit Non-Computable Solutions}

\author{Holger~Boche,~\IEEEmembership{Fellow, IEEE} and Volker~Pohl,~\IEEEmembership{Member, IEEE}
and H.~Vincent~Poor,~\IEEEmembership{Life Fellow, IEEE} and Frank~H.P.~Fitzek,~\IEEEmembership{Fellow, IEEE}
\thanks{
H.~Boche is with the Chair of Theoretical Information Technology, Technical University of Munich, Arcisstra{\ss}e 21, 80333 M{\"u}nchen, Germany, e-mail: boche@tum.de.
He is also
with the Munich Center for Quantum Science and Technology (MCQST), Munich, Germany, and with the Munich Quantum Valley (MQV), Munich, Germany.
V.~Pohl is with the Chair of Theoretical Information Technology, Technical University of Munich, Arcisstra{\ss}e 21, 80333 M{\"u}nchen, Germany, e-mail: volker.pohl@tum.de.
H.~V.~Poor is with the Department of Electrical and Computer Engineering, Princeton University, Princeton, NJ~08544, USA, e-mail: poor@princeton.edu.
F.H.P.~Fitzek is with the Deutsche Telekom Chair for Communication Networks, Dresden University of Technology, Dresden, Germany, e-mail: frank.fitzek@tu-dresden.de 
}
\thanks{
This work was partly supported by the German Federal Ministry of Research, Technology and Space (BMFTR) within the national initiative on 6G Communication Systems through the transfer hub \emph{6G-life} under Grants 16KIS2413K and 16KIS2414,
by the German Research Foundation (DFG) as part of Germany's Excellence Strategy -- EXC 2050/2 -- Project ID 390696704 -- Cluster of Excellence "Centre for Tactile Internet with Human-in-the-Loop" (CeTI) of Technische Universit{\"a}t Dresden,
and by the Bavarian State Ministry of Science and the Arts within the project \emph{Next Generation AI Computing (GAIn)}.
The work of H.~V.~Poor was partly supported by Princeton Language and Intelligence.
The paper was also supported by the LTI TUM collaboration with Princeton through the BMFTR as part of the 6G-Atlantic Bridges project.}
}


\maketitle
\begin{abstract}
If a regular linear differential-algebraic equation has a unique solution, this solution is known to be Turing computable.
In contrast, this paper shows that singular linear differential-algebraic equations exhibit fundamentally different computability properties.
More precisely, we show that for every linear differential-algebraic equation (LDAE) with constant computable coefficients, there exists a computable input function with strong regularity properties such that the LDAE admits a unique solution, yet this solution is not Turing computable.
In fact, the solution is non-computable at certain computable times.
\end{abstract}
\begin{IEEEkeywords}
Differential-algebraic equations, Turing computability
\end{IEEEkeywords}

\section{Introduction}
\label{sec:intro}

In this paper, we investigate \emph{linear differential-algebraic equations (LDAE)} with constant coefficients of the form
\begin{align}
\label{equ:LDAE}
\begin{split}
	\mE \vx'(t) &= \mA \vx(t) + \vu(t)\,,
	\qquad t\in\TT\;,\\
	\text{with}\qquad
	\vx(0) &= \vx_{0}\;,
\end{split}
\end{align}
on intervals $\TT = [0,T] \subset \RN$.
Here $\mE,\mA \in \RN^{m\times n}$ are coefficient matrices and the vector $\vx_{0} \in \RN^{n}$ contains the initial values of the LDAE.
The vector valued function $\vu : \TT \to \RN^{m}$ is often called the \emph{input} of the LDAE.
For given initial values $\vx_{0}\in\RN^{n}$, we call $\vu$ \emph{admissible} if \eqref{equ:LDAE} admits at least one solution.
Throughout this paper, we assume that the coefficient matrices $\mE$ and $\mA$ and the vector $\vx_{0}$ are \emph{computable}, meaning that all of their entries are computable real numbers.

LDAEs are a widely accepted tool for modeling constrained dynamical systems \cite{KungelMehrmann_DAEs} in numerous applications, such as power systems, nonlinear circuits or control systems \cite{HillMareels_TCS90,Maheshawari_TCS21,Sau_TCAS18,Schon_CDC2003}.
In general, they can be viewed as a combination of an ordinary differential equation (ODE) with some algebraic equations relating the single components of the independent variable $\vx$. 
LDAEs usually arise when a physical process is modeled by an ODE whose state variables are subject to additional constraints.
A typical example of such an algebraic constraint arises when a mass point, whose motion is described by Newton's ODEs, is restricted to a certain subspace.
Another typical example arises from Kirchhoff's circuit laws, which restrict the voltages and currents in a electrical network to follow certain rules (cf. also the example in Section~\ref{sec:ExampleKirchhoff}).

Only for a limited number of LDAEs, it is possible to derive closed form expressions for the solution.
The vast majority of LDAEs that appear in practical applications require numerical methods to compute an approximation of the solution.
Since digital computers are the primary platform for such computations, a fundamental question is whether solutions of LDAEs are generally Turing computable \cite{Turing_1937,Turing_1938}.
The central question is whether it is possible to compute an approximation of the solution under an effective control of the approximation error.
As we will show, the answer to this question depends strongly on the structure of the underlying LDAE. 


In general, one distinguishes two different cases:
\emph{a) Regular LDAEs:} If the coefficient matrix $\mE$ in \eqref{equ:LDAE} is non-singular, then we can multiply the first line of \eqref{equ:LDAE} with $\mE^{-1}$ to obtain an ODE with constant coefficients:
\begin{align*}
	\vx'(t) &= \widetilde{\mA}\vx(t) + \widetilde{\vu}(t)\,,
	\qquad t\in\TT\,,\\
	\text{with}\qquad \vx(0) &= \vx_{0}\;,
\end{align*}
and with the coefficient matrix $\widetilde{\mA} = \mE^{-1}\mA$ and input vector $\widetilde{\vu}(t) = \mE^{-1}\vu(t)$.
In this case, the answer to the question whether the solution is Turing computable is already well known:
Assume all initial values in $\vx_{0}$ are computable numbers and $\vu$ is an input vector of computable continuous functions. 
Then the ODE has a unique solution $\vx$ and because of its uniqueness, $\vx$ is always a computable continuous function \cite{Graca_CompDGL_2021}.
Moreover, there exists many efficient and well established numerical algorithms to solve ODEs with a unique solution.

\emph{b) Singular LDAEs:}
The situation changes dramatically if $\mE$ is singular.
Then \eqref{equ:LDAE} can no longer be converted to an ODE and the fundamental question whether its solution is Turing computable seems to be open.
However, one has already observed that numerical methods that work good for solving regular LDAEs often fail if they are applied to singular LDAEs. \cite{Petzold_DAEnotODE_SIAM82}.

So the aim of this paper is to study the Turing computability of solutions of \emph{singular LDAEs}.
We will show that these solutions exhibit fundamentally different computability properties from those of ODEs.
To this end, we consider singular LDAEs \eqref{equ:LDAE} with computable coefficient matrices $\mE$ and $\mA$ and with computable initial values $\vx_{0}$. 
Then we consider only computable input vectors $\vu$ such that \eqref{equ:LDAE} possesses a unique solution.
For these kind of equations we ask the following
\begin{question}
Is it true that, as in the case of ODEs, the unique solution of a singular LDAE is always Turing computable?
\end{question}
We will show in this paper that the answer is negative.
In fact, we are going to show that one can always find a $\vu$ that has additional strong regularity properties (e.g. smoothness with computable derivatives, finite Dirichlet energy, absolute continuous) such that the singular LDAE
\eqref{equ:LDAE} has a unique solution $\vx$ but such that this solution is not a computable continuous function.
Even more, we will show that on any arbitrary small interval $[0,T] \subset \RN$ there exists computable points $\tau\in\TT$ at which the solution $\vx(\tau)$ is not computable.
Thus the solution does not even satisfies the weakest form of computability, namely Banach-Mazur computability.

The remainder of the paper has the following structure:
Section~\ref{sec:Notation} introduces the main notation and function spaces used in this paper, and
Section~\ref{sec:Computability} shortly recalls notation and concepts from computability analysis.
Section~\ref{sec:LDAE} provides the necessary background on the theory of singular LDAEs.
Section~\ref{sec:MainResults} then presents our main results and shows that the solutions of singular LDAEs are generally not computable.
Section~\ref{sec:Proof_AuxResultsAndMainLemma} contains the long and technical proof of the fundamental Lemma~\ref{lem:MainLemma}.

\section{General Notation}
\label{sec:Notation}

Vectors in  $\RN^{n}$ will be denoted by boldface lower case letters.
For example, $\va \in \RN^{n}$ denotes a column vector $\va = (a_{1}, a_{2}, \dots, a_{N})^{\T}$
with entries $a_{k}\in\RN$, $k = 1,\dots, n$.
Matrices are denoted by boldface upper case letters. For example $\mA \in \CN^{m,n}$ denotes a matrix with $m$ rows $n$ columns and entries $a_{\ell,k}$, $\ell = 1,\dots m$ and $k = 1,\dots n$.
We recall that a square matrix $\mA \in \CN^{n\times n}$ is said to be \emph{nilpotent} if there exists a positive integer $m\in\NN$ such that $\mA^{m} = 0$.
The smallest integer $m$ that satisfies this equation is said to be the \emph{index of nilpotency} of $\mA$.
Note that the index of nilpotency of a $n\times n$ matrix is always smaller or equal to $n$.

Let $\TT = [0,T] \subset \RN$ be an arbitrary closed interval on the real axis $\RN$ starting at zero.
Then $\C(\TT)$ stands for the Banach space of continuous functions on $\TT$ with the usual norm
\begin{equation*}
	\left\|f\right\|_{\infty}
	= \max_{t\in\TT} \left|f(t)\right|\,.
\end{equation*}
The set of $k$-times continuously differentiable functions on $\TT$ is denoted by $\C^{k}(\TT)$, and we write $f^{(\ell)}$, $\ell \leq k$, for the $\ell$-th derivative of $f \in \C^{k}(\TT)$.
For $1\leq p < \infty$, $L^{p}(\TT)$ denotes the usual Banach space of integrable functions on $\TT$ with norm
$\left\|f\right\|_{p} = \left( \int_{\TT} \left|f(t)\right|\, \d t \right)^{1/p}$.
A function $f\in\C(\TT)$ is said to be \emph{absolute continuous}, if there exists a $g\in L^{1}(\TT)$ such that $f(t) = f(0) + \int^{t}_{0} g(\tau)\, \d\tau$ for all $t\in\TT$.

Spaces of vector-valued function on $\TT$, i.e. $n$-dimensional vectors whose entries are functions in a certain function space are denoted by $\C(\TT,\CN^{n})$, $\C^{k}(\TT,\CN^{n})$, $L^{p}(\TT,\CN^{n})$, etc.
The corresponding functions are also denoted by lower-case boldface letters, e.g. $\vx\in \C(\TT,\CN^{n})$ is a vector
\begin{equation*}
	\vx(t) = \left(x_{1}(t), \dots, x_{n}(t) \right)^{\T}\,,
	\quad t\in\TT\,,
\end{equation*}
whose entries $x_{k}$, $1\leq k \leq n$ are continuous functions on $\TT$.

For any $f\in L^{1}(\TT)$, its \emph{Fourier coefficients} are given by
\begin{equation}
\label{equ:FourierCoef}
	a_{n}(f) = \frac{2}{T}\int^{T}_{0} f(t)\, \cos\big(\tfrac{2\pi}{T} n t\big)\, \d t\,,
	\qquad\text{and}\qquad
	b_{n}(f) = \frac{2}{T}\int^{T}_{0} f(t)\, \sin\big(\tfrac{2\pi}{T} n t\big)\, \d t\,,\qquad n=0,1,2,\dots,
\end{equation}
and any $f\in L^{2}(\TT) \subset L^{1}(\TT)$ can be recovered from its Fourier coefficients by the harmonic series
\begin{equation}
\label{equ:FourierSeries}
	f(t)
	= \frac{a_{0}(f)}{2} + \sum^{\infty}_{n=1} \left[ a_{n}(f)\cos\big(\tfrac{2\pi}{T} n t\big) + b_{n}(f)\sin\big(\tfrac{2\pi}{T} n t\big) \right],
\end{equation}
where the sum converges in the norm of $L^{2}(\TT)$ and equality in \eqref{equ:FourierSeries} means equality in $L^{2}(\TT)$.
The \emph{Wiener algebra}, denoted by $\W(\TT)$, is the space of all absolutely convergent Fourier series on $\TT$ with norm
\begin{equation*}
	\left\|f\right\|_{\W(\TT)} = \textstyle\frac{\left| a_{0}(f) \right|}{2} + \sum_{n\in\NN} \left|a_{n}(f) \right| + \left|b_{n}(f) \right| < \infty\;.
\end{equation*}
The absolute summability of the Fourier series guarantees its uniform convergence and so one has $\A(\TT) \subset \C(\TT)$.

According to Parseval's identity, one has
\begin{equation*}
	\left\|f\right\|^{2}_{2}
	= \frac{\left|a_{0}(f)\right|^{2}}{4} + \frac{1}{2}\sum^{\infty}_{n=0} \left( \left|a_{n}(f)\right|^{2} + \left|b_{n}(f)\right|^{2} \right)\,.
\end{equation*}
Often, this value is called the energy of $f$ and so $L^{2}(\TT)$ is said to be the space of functions on $\TT$ with finite energy.
Nevertheless, from a physical point of view the so-called \emph{Dirichlet energy} \cite{Ross_DirichletSpace,Arcozzi_DirichletSpace11} is often of much more physical importance.
For $f \in L^{2}(\TT)$, it is given by
\begin{equation*}
	\Energy(f) 
	= \frac{1}{2}\sum^{\infty}_{n=1} n \left( \left|a_{n}(f)\right|^{2} + \left|b_{n}(f)\right|^{2} \right)\;,
\end{equation*}
and we write
\begin{equation*}
	H^{1/2}(\TT) = \left\{ f \in L^{2}(\TT)\ :\ \Energy(f) < \infty \right\}
\end{equation*}	
for the Banach space of all $L^{2}(\TT)$-functions of finite Dirichlet energy,
with the norm
\begin{equation}
\label{equ:H12_norm}
	\left\|f\right\|_{H^{1/2}}
	= \left( \tfrac{\left|a_{0}(f)\right|^{2}}{4} + \Energy(f) \right)^{1/2}\,.
\end{equation}
Note that $H^{1/2}(\TT) = W^{\frac{1}{2},2}(\TT)$ is the usual Sobolev-space with index $1/2$ \cite{Adams_SobolevSpaces,Evans_PDEs}.

Functions of the form \eqref{equ:FourierSeries} but where the sum contains only finitely many non-zero terms are called \emph{trigonometric polynomials} on $\TT$.

\section{Computability analysis}
\label{sec:Computability}

We briefly recall some of its notions and concepts from computability analysis \cite{PourEl_Computability,Weihrauch_ComputableAnalysis}. 
Its theory is based on the notion of a \emph{classical Turing machine}, which is an algorithmic device that receives rational numbers as input and, according to a fixed computational procedure, yields rational numbers as output. 

\subsection{Computable numbers}

Classical Turing machines only operate with rational numbers.
Real numbers needs to be approximated \emph{effectively} by rational numbers:

\begin{definition}
\label{def:CompNumber}
An $x\in\RN$ is said to be \emph{computable} if there exists a classical Turing machine $\TM_{x} : \NN\to\QN$ such that
\begin{equation}
\label{equ:DefCompNumb}
	\left|x-\TM_{x}(M)\right| \leq 2^{-M},
	\qquad\text{for all}\ M\in\NN\;.
\end{equation}
$\RNc$ denotes the closed subfield of all computable real numbers.
\end{definition}
Effective convergence means that one is able to control the approximation error: 
To every $M\in\NN$ one can algorithmically determine an approximation $r_{M} = \TM_{x}(M)\in\QN$ such that the approximation error satisfies \eqref{equ:DefCompNumb}.

\begin{remark}
The term $2^{-M}$ in relation \eqref{equ:DefCompNumb} of Definition~\ref{def:CompNumber} (and in all subsequent definitions involving effective convergence) may be replaced by any functions $r : \NN \to \RN$ such that $r(M)$ monotonically decreases to zero as $M \to\infty$.
\end{remark}

In a similar way, one defines \emph{computable sequences}:
\begin{definition}
A sequence $\bsx = \left\{ x_{n} \right\}_{n\in\NN}\subset\RN$ is said to be \emph{computable} if there exists a single classical Turing machine $\TM_{\bsx} : \NN\times\NN \to\QN$ such that for every $n\in\NN$
\begin{equation*}
	\left|x_{n} - \TM_{\bsx}(n,M)\right| \leq 2^{-M},
	\quad\text{for all}\ M\in\NN\,.
\end{equation*}
\end{definition}

\begin{remark}
Note that, according to the previous definition, every element $x_{n}$ of a computable sequence $\bsx$ is a computable number.
Moreover, the previous definition can easily be extended to sequences with more than one index.
The important point is always, that there exists a \emph{single} classical Turing machine that es able to effectively compute every element of the sequence.
\end{remark}

\subsection{Semi-computable numbers}

For a non-computable $x\in\RN$, one may ask to which extent it is not computable.
\emph{Zheng} and \emph{Weihrauch} introduced a hierarchy of non-computable numbers \cite{ZhengWeihrauch_MathLog01} that characterizes the degree of non-computability.
Some sets from the first layer of this hierarchy are needed in this paper.
Following \cite{ZhengWeihrauch_MathLog01}, we define two sets of \emph{semi-computable} numbers:

\begin{definition}[Semi-computable numbers] 
An $x \in\RN$ is said to be \emph{left-computable}
if there exists a computable sequence $\left\{ \zeta_{n} \right\}_{n\in\NN} \subset\RNc$ such that
\begin{equation*}
	\zeta_{n+1} \geq \zeta_{n}\,,\ \text{for all}\ n\in\NN\,,
	\quad\text{and}\quad
	\lim_{n\to\infty} \zeta_{n} = x\;.
\end{equation*}
An $x \in\RN$ is said to be \emph{right-computable}
if there exists a computable sequence $\left\{ \xi_{n} \right\}_{n\in\NN}\subset\RNc$ such that
\begin{equation}
\label{equ:propXi}
	\xi_{n+1} \leq \xi_{n}\,,\ \text{for all}\ n\in\NN\,,
	\quad\text{and}\quad
	\lim_{n\to\infty} \xi_{n} = x\;.
\end{equation}
We denote by $\Pu$ and $\Pd$ the set of all left- and right-computable numbers, respectively,
and $\Csemi = \Sigma_{1} \cup \Pi_{1}$ is the set of all semi-computable numbers.
\end{definition}
A semi-computable number is generally not computable, i.e. $\RNc \subsetneq \Pu$ and $\RNc \subsetneq \Pd$.
However, it is not hard to see \cite{PourEl_Computability} that if $x\in\RN$ is left- \emph{and} right-computable, then it is also computable in the sense of Def.~\ref{def:CompNumber}.

\begin{proposition}
\label{prop:Delta1isUnion}
A number $x \in\RN$ belongs to $\RNc$ if and only if $x \in \Sigma_{1} \cap \Pi_{1}$, that is $\RNc = \Sigma_{1} \cap \Pi_{1}$.
\end{proposition}
$\RNc$ is closed under the usual arithmetic operations addition, subtraction, multiplication and division.
This is not true for $\Sigma_{1}$, $\Pi_{1}$, and $\Csemi$.
In fact, there exist left computable numbers $a,b\in\Sigma_{1}$ such that $x = a-b \notin \Csemi$.
This motivated \cite{WeihrauchZhengHierarchy_98} to define the following set of \emph{weakly computable numbers}:

\begin{definition}[Weakly computable numbers]
\label{def:WeakComputableNr}
An $x \in\RN$ is said to be weakly computable if there exist $a,b\in\Sigma_{1}$ such that $x = a-b$.
The set of all weakly computable numbers is denoted by $\Cweak$.
\end{definition}

\begin{remark}
$\RNc$, $\Sigma_{1}$, $\Pi_{1}$, and $\Csemi$ are proper subsets of $\Cweak$.
\end{remark}
For us, it will be important that there exist numbers $x \in \Cweak$ that are not computable.

\subsection{Computable functions}

We consider functions on closed intervals $\TT = [0,T]$  with $T > 0$.
Such an interval is said to be computable if $T\in\RNc$.
There are different notions that describe the computability of such functions.
The weakest and most natural of these notions is the so called Banach--Mazur computability.

\begin{definition}[Banach--Mazur computable function]
\label{def:BMCompFunc}
Let $\TT = [0,T]$ be a computable interval. A function $f : \TT\to\RN$ is called \emph{Banach--Mazur computable} (or \emph{sequentially computable}) if,
for every computable sequence $\left\{ t_{n} \right\}_{n\in\NN} \subset \TT$ the sequence $\left\{ f(t_{n}) \right\}_{n\in\NN}$ is computable.
\end{definition}

\begin{remark}
In particular, a Banach-Mazur computable function maps computable numbers onto computable numbers.
\end{remark}

Banach-Mazur computability is the weakest form of computability.
It does not impose any continuity requirements of $f$ on the entire set $\TT$.
Thus a Banach-Mazur computable function might not be continuous on $\TT$.
In particular, if the sequence $\left\{ t_{n} \right\}_{n\in\NN}$ effectively converges to $t\in\TT$ then this does not necessarily imply that also $\left\{ f(t_{n}) \right\}_{n\in\NN}$ converges effectively.

For us, the notion of a computable continuous function will be of more importance. 
This form of computability can be defined by an effective version of Weierstrass approximation theorem.
To this end, we say that a sequence $\left\{ \varphi_{n} \right\}_{n\in\NN}$ is a \emph{computable sequence of trigonometric polynomials} on $\TT$ if the functions $\varphi_{n}$ have the form
\begin{equation}
\label{equ:TrigPoly}
	\varphi_{n}(t)
	= \tfrac{a_{n}(0)}{2} + \sum^{\TM(n)}_{k=1} \!\!\!\left[ a_{k}(n)\, \cos\left(\tfrac{2\pi}{T}k t\right) + b_{k}(n)\, \sin\left(\tfrac{2\pi}{T}k t\right) \right],	
\end{equation}
for all $t\in\TT$ and
wherein $\left\{ a_{k}(n) \right\}_{n,k\in\NN}, \left\{ b_{k}(n) \right\}_{n,k\in\NN}$ are computable (double) sequences and $\TM : \NN\to \NN$ is classical Turing machine that computes for every $n\in\NN$ the degree $\TM(n)$ of the trigonometric polynomial $\varphi_{n}$.

Let $\B$ be an arbitrary Banach spaces of functions on $\TT$.
A function $f \in\B$ is said to be computable (as a function of $\B$) if it can effectively be approximated by computable trigonometric polynomials in the following sense:

\begin{definition}[Computable functions in Banach spaces]
\label{def:ComputableBanachSpace}
Let $\B$ be a Banach space of functions on $\TT$.
1) A function $f : \B \to\RN$ is called \emph{$\B$-computable} if there exists a computable sequence $\left\{ \varphi_{n} \right\}_{n\in\NN}$ of trigonometric polynomials that effectively converges to $f$,
i.e. there exists a classical Turing machine $\TM : \NN\to\NN$ such that
\begin{equation*}
	n\geq \TM(M)
	\quad\text{implies}\quad
	\left\|f - \varphi_{n} \right\|_{\B} < 2^{-M}\,.
\end{equation*}
2) A sequence $\left\{ f_{n} \right\}_{n\in\NN} \subset \B$ is said to be $\B$-computable if there exists a computable double sequence $\left\{ \varphi_{n,k} \right\}_{n,k\in\NN}$ of trigonometric polynomials and classical Turing machine $\TM : \NN\times\NN\to\NN$ such that
\begin{equation*}
	k \geq \TM(n,N)
	\quad\text{implies}\quad
	\left\|\varphi_{n,k} - f_{n} \right\|_{\B} \leq 2^{-N}\,.
\end{equation*}
The set of all $\B$-computable function is denoted by $\Bc$.
\end{definition}
Subsequently, we need in particular the two space $\Cc(\dD)$ and $H^{1/2}_{\mathrm{c}}(\dD)$.
A function $f \in \Cc(\dD)$ is called a \emph{computable continuous function}.

\begin{remark}
Note that if $\big\{ f_{n} \big\}_{n\in\NN} \subset\B$ is a computable sequence that effectively converges to $f$ in the sense of Def.~\ref{def:ComputableBanachSpace} then $\big\{ \left\|f_{n}\right\|_{\B} \big\}_{n\in\NN}$ is a computable sequence that effectively converges to $\left\|f\right\|_{\B}$ \cite{PourEl_Computability}.
In particular, $\left\|f\right\|_{\B} \in\RNc$ for every computable $f \in\B$.
\end{remark}

\begin{remark}
Note that every computable continuous function is also Banach--Mazur computable \cite{PourEl_Computability}. 
Conversely there are Banach--Mazur computable functions that are not computable continuous functions.
\end{remark}

\section{Singular LDAEs}
\label{sec:LDAE}

We briefly recall the general structure of LDAEs \eqref{equ:LDAE} and properties of their solutions.

\subsection{Canonical form of LDAEs}

If $\mP\in\RN^{m\times m}$ and $\mQ\in\RN^{n\times n}$ are two non-singular matrices then we can transform \eqref{equ:LDAE} into an equivalent form by changing the coordinates in $\RN^{m}$ and $\RN^{n}$ according to
\begin{equation}
\label{equ:CoordinateTranfor}
	\vx(t) = \mQ\, \widetilde{\vx}(t)
	\qquad\text{and}\qquad
	\vu(t) = \mP\, \widetilde{\vu}(t)\,,
	\quad\text{for all}\ t\in\TT\,.
\end{equation}
This yields an equivalent LDAE
\begin{align}
\label{equ:LDAE_tilde}
\begin{split}
	\widetilde{\mE} \widetilde{\vx}'(t) &= \widetilde{\mA} \widetilde{\vx}(t) + \widetilde{\vu}(t)\,,
	\qquad t\in\RN\;,\\
	\text{with}\qquad
	\widetilde{\vx}(0) &= \vx_{0} = \widetilde{\vx}_{0}\;,
	\end{split}
\end{align}
with $\widetilde{\vx}_{0} = \mQ^{-1}\, \vx_{0}$ and with the coefficient matrices
\begin{equation}
\label{equ:TrafoMatrices}
	\widetilde{\mE} = \mP^{-1}\mE\mQ
	\qquad\text{and}\qquad
	\widetilde{\mA} = \mP^{-1}\mA\mQ\,.
\end{equation} 
It is clear that instead of \eqref{equ:LDAE}, we can always consider the transformed LDAE \eqref{equ:LDAE_tilde}, because of the bijective mappings \eqref{equ:CoordinateTranfor} between the solutions of both equations.

\begin{definition}[Strong equivalence]
\label{equ:StrongEquivalence}
Two pairs of matrices $(\mE,\mA)$ and $(\widetilde{\mE},\widetilde{\mA})$ with $\mE,\mA,\widetilde{\mE},\widetilde{\mA} \in \RN^{m\times n}$
are said to be \emph{strongly equivalent} if there exist non-singular matrices $\mP\in\RN^{m\times m}$ and $\mQ\in\RN^{n\times n}$ such that \eqref{equ:TrafoMatrices} holds.
\end{definition}

\begin{remark}
Note that the order of the matrices matters, i.e. $(\mE,\mA)$ and $(\mA,\mE)$ are generally different matrix pairs.
\end{remark}

\begin{remark}
Pairs $(\mE,\mA)$ of matrices of the same size are also known as \emph{matrix pencils}. Often they are associated with the the matrix polynomial $\lambda\mE - \mA$ or with the matrix polynomial $\lambda\mE - \mu\mA$ in two independent variables \cite{Gantmacher1}.
\end{remark}
Definition~\ref{equ:StrongEquivalence} establishes an equivalence relation on the set of all matrix pairs of the same size.
We only consider LDAEs \eqref{equ:LDAE} for which the pair $(\mE,\mA)$ is regular, in the following sense:

\begin{definition}[Regular matrix pairs]
A matrix pair $(\mE,\mA)$ with $\mE,\mA \in \RN^{m\times n}$ is said to be \emph{regular} if $m=n$ and if the corresponding \emph{characteristic polynomial}
$p(\lambda) = \det(\lambda\mE - \mA)$
is not the zero polynomial.
\end{definition}

\begin{remark}
Note that even if the matrix pair $(\mE,\mA)$ is regular, one or even both matrices of this pair might be singular.
In particular, there exists \emph{singular} LDAEs with a \emph{regular} pair $(\mE,\mA)$ of coefficient matrices.
In fact, this paper focuses mainly on such LDAEs.
\end{remark}
It is easy to see that if a matrix pair $(\mE,\mA)$ is strongly equivalent to a regular matrix pair, then $(\mE,\mA)$ itself is regular.

Having an equivalence relation as given by Def.~\ref{equ:StrongEquivalence}, one may ask for particular simple canonical forms.
Based on works of Weierstrass and Kronecker, Gantmacher \cite{Gantmacher1} has given a complete characterization of general matrix pairs together with a corresponding so called \emph{Kronecker canonical form}.

Since we consider only regular matrix pairs, the following canonical form will be sufficient for us.

\begin{theorem}
\label{thm:CanonicalRegularPencil}
Every regular matrix pair $(\mE,\mA)$ with $\mE,\mA \in \RN^{n\times n}$ is strongly equivalent to a matrix pair of the form
\begin{equation}
\label{equ:MazrixWeierstrassCF}
	\big( \widetilde{\mE},\widetilde{\mA} \big)
	= \left(\left[\begin{array}{cc}\mI_{n_\mathrm{a}} & \mNull\\ \mNull &\mN\end{array}\right],\left[\begin{array}{cc}\mJ & \mNull\\ \mNull & \mI_{n_{\mathrm{b}}} \end{array}\right]\right)
\end{equation}
where $\mI_{n_{\mathrm{a}}}$ and $\mI_{n_{\mathrm{b}}}$ are identity matrices of size $n_{\mathrm{a}}\times n_{\mathrm{a}}$ and $n_{\mathrm{b}}\times n_{\mathrm{ab}}$, respectively, 
where $\mJ \in \RN^{n_{\mathrm{a}}\times n_{\mathrm{a}}}$ is regular matrix in Jordan canonical form, 
where $\mN \in \RN^{n_{\mathrm{b}} \times n_{\mathrm{b}}}$ is a nilpotent matrix in Jordan canonical form,
and with $n_{\mathrm{a}} + n_{\mathrm{b}} = n$.
\end{theorem}

\begin{remark}
Note that $\na$ or $\nb$ might be zero.
If $\nb = 0$ then $\widetilde{\mE}$ and $\mE$ are regular matrices. In this case \eqref{equ:LDAE} is a \emph{regular LDAE}.
If $\nb > 0$ then $\widetilde{\mE}$ and $\mE$ are singular and the corresponding LDAE \eqref{equ:LDAE} is a \emph{singular LDAE}.
\end{remark}
The lower diagonal block $\mN$ of the matrix $\widetilde{\mE}$ in the canonical form \eqref{equ:MazrixWeierstrassCF} is a nilpotent matrix.
Its index of nilpotency plays a particular role in characterizing solutions of singular LDAEs:

\begin{definition}[Index of an LDAE]
Let $(\mE,\mA)$ be a regular matrix pair associated with an LDAE \eqref{equ:LDAE}.
The index of nilpotency of the matrix $\mN$ in the canonical form \eqref{equ:MazrixWeierstrassCF} is called the \emph{index of the LDAE} \eqref{equ:LDAE} and is denoted by $\ind(\mE,\mA)$.
\end{definition}

\begin{remark}
Thus $\ind(\mE,\mA) = 0$ if and only if the corresponding LDAE \eqref{equ:LDAE} is regular, otherwise $\ind(\mE,\mA) \geq 1$.
\end{remark}

Throughout the remainder of this paper, we always assume that the coefficients of \eqref{equ:LDAE}, i.e. the entries of the matrices $\mE$ and $\mA$ and of the vector $\vx_{0}$ are computable numbers:

\begin{definition}[Computable matrices and vectors]
A matrix $\mM \in \RN^{m\times n}$ or a vector $\vv\in\RN^{n}$ is said to be a \emph{computable matrix} or a \emph{computable vector}, respectively, if all of its entries are computable real numbers.
\end{definition}

\begin{remark}
If $\mE$ and $\mA$ are computable matrices then the characteristic polynomial $p(\lambda)$ is a computable polynomial (a polynomial whose coefficients are computable numbers) of degree $r\leq n$.
\end{remark}
If $\mE$ and $\mA$ are computable matrices then also the coefficient matrices $\widetilde{\mE}$ and $\widetilde{\mA}$ of the canonical form \eqref{equ:MazrixWeierstrassCF} and the transformation matrices $\mP$ and $\mQ$ in \eqref{equ:CoordinateTranfor} are computable matrices.
This is shown by the following lemma.

\begin{lemma}
\label{lem:Computable_PQ}
Let $(\mE,\mA)$ be a regular pair of computable matrices, and let $(\widetilde{\mE},\widetilde{\mA})$ be the corresponding canonical form \eqref{equ:MazrixWeierstrassCF}.
Then the matrices $\widetilde{\mE}$, $\widetilde{\mA}$ and the matrices $\mP$ and $\mQ$ which transform $(\mE,\mA)$ onto $(\widetilde{\mE},\widetilde{\mA})$ via \eqref{equ:TrafoMatrices} are computable matrices.
\end{lemma}
\begin{remark}
A sketch of proof of this statement is given in Appendix~\ref{app:ComputablePQ}.
\end{remark}


If $(\mE,\mA)$ is a regular matrix pair with $\mE,\mA \in \RN^{n\times n}$, then $p(\lambda) = \det(\lambda\mE - \mA)$ is a non-zero polynomial in $\lambda$ of degree $r \leq n$.
A scalar $\lambda\in\CN$ is said to be an \emph{eigenvalue} of $(\mE,\mA)$ if $p(\lambda) = 0$.
If $r < n$ then $(\mE,\mA)$ has necessarily at least one "eigenvalue at infinity".
With these notation, the first diagonal block in the canonical form \eqref{equ:MazrixWeierstrassCF} is related to the finite eigenvalues of the regular $(\mE,\mA)$, whereas the second diagonal block is related to the eigenvalues at infinity of $(\mE,\mA)$.
If $(\mE,\mA)$ has only finite eigenvalues, the second block in \eqref{equ:MazrixWeierstrassCF} does not exit, and $n = \na$. In this case \eqref{equ:LDAE} represents a regular LDAE.
If $(\mE,\mA)$ has at least one eigenvalue at infinity then $\mE$ is singular and \eqref{equ:LDAE} represents a singular LDAE.

\begin{remark}
To distinguish the finite and infinite eigenvalues of the matrix pair $(\mE,\mA)$, it is often more convenient to consider the polynomial $p(\lambda,\mu) = \det(\lambda\mE - \mu\mA)$ in two independent variables $\lambda$ and $\mu$.
For a regular pair $(\mE,\mA)$, this polynomial has always degree $n$ and can be factorized as
\begin{equation*}
	p(\lambda,\mu) = C_{0} \prod_{k} \left(\alpha_{k}\lambda - \mu\right)^{n_k}
\end{equation*}
with non-zero constants $C_{0}\in\RN$ and $\alpha_{k} \in\CN$ and with  $\sum_{k} n_k = n$.
Factors $\left(\alpha_{k}\lambda - \mu\right)^{n_k}$ with $\alpha_{k} \neq 0$ correspond to the finite eigenvalues $\lambda = 1/\alpha_{k}$ of $(\mE,\mA)$, whereas factors with $\alpha_{k}=0$ correspond to the eigenvalues at infinity.
\end{remark}

\subsection{General solution of LDAEs}

According to Theorem~\ref{thm:CanonicalRegularPencil} every singular LDAE \eqref{equ:LDAE} with a regular pair $(\mE,\mA)$ of coefficient matrices can be transformed to the following system of equations
\begin{align}
\label{equ:WeierstrassCanonForm}
\begin{split}
	\widetilde{\vx}'_{\mathrm{a}}(t) &= \mJ \widetilde{\vx}_{\mathrm{a}}(t) + \widetilde{\vu}_{\mathrm{a}}(t)\qquad \text{with}\qquad \widetilde{\vx}_{\mathrm{a}}(0) = \widetilde{\vx}_{\mathrm{a},0}\\
	\mN \widetilde{\vx}'_{\mathrm{b}}(t) &= \phantom{\mJ}\widetilde{\vx}_{\mathrm{b}}(t) + \widetilde{\vu}_{\mathrm{b}}(t)\qquad \text{with}\qquad \widetilde{\vx}_{\mathrm{b}}(0) = \widetilde{\vx}_{\mathrm{b},0}.
\end{split}	
\end{align}
with 
\begin{align}
\begin{split}
\label{equ:TrafoToCanonicalForm}
	\vx(t)
	&= \mQ \begin{pmatrix}
	\widetilde{\vx}_{\mathrm{a}}(t) \\ \widetilde{\vx}_{\mathrm{b}}(t)
	\end{pmatrix}
	= \mQ\, \widetilde{\vx}(t),	\quad t\in \TT\,,\\
	\vu(t)
	&= \mP \begin{pmatrix}
	\widetilde{\vu}_{\mathrm{a}}(t) \\ \widetilde{\vu}_{\mathrm{b}}(t)
	\end{pmatrix}
	= \mP\, \widetilde{\vu}(t), \quad t\in \TT\,,\\
	\vx_{0}
	&= \mQ \begin{pmatrix}
	\widetilde{\vx}_{\mathrm{a},0} \\ \widetilde{\vx}_{\mathrm{b},0}
	\end{pmatrix}
	= \mQ\, \widetilde{\vx}_{0}\,.
\end{split}
\end{align}

\paragraph{The regular part}
The first line of \eqref{equ:WeierstrassCanonForm} corresponds to a system of first order ODEs.
If $\widetilde{\vu}_{\mathrm{a}}\in \C(\TT,\RN^{n_{\mathrm{a}}})$ is a continuous function then, by the Picard–Lindelöf theorem, this ODE
has a unique solution $\widetilde{\vx}_{\mathrm{a}} \in \C^{1}(\TT,\RN^{\na})$ for arbitrary initial values $\widetilde{\vx}_{\mathrm{a},0}\in\RN^{\na}$.
The uniqueness of the solution then implies that $\widetilde{\vx}_{\mathrm{a}}$ is always computable function on $\TT$ \cite{Graca_CompDGL_2021}.

\paragraph{The singular part}
The second line of \eqref{equ:WeierstrassCanonForm} has the form
\begin{subequations}
\label{equ:OurEquation}
\begin{align}
\label{equ:OurEquation_a}
	\mN\, \widetilde{\vx}'_{\mathrm{b}}(t)
	&= \widetilde{\vx}_{\mathrm{b}}(t) + \widetilde{\vu}_{\mathrm{b}}(t)\,,
	\quad\ t\in\TT\,,\\
\label{equ:OurEquation_b}	
	\text{with}\qquad
	\widetilde{\vx}_{\mathrm{b}}(0) &= \widetilde{\vx}_{_{\mathrm{b}},0}\,,
\end{align}
\end{subequations}
wherein $\mN \in \RN^{n_{\mathrm{b}} \times n_{\mathrm{b}}}$ is a nilpotent matrix in Jordan canonical form with index of nilpotency $m \leq n_{\mathrm{b}}$,
wherein $\widetilde{\vu}_{\mathrm{b}} \in \C^{m}(\TT,\RN^{n_{\mathrm{b}}})$ is a given input function which is assumed to be sufficiently often continuously differentiable,
and where $\widetilde{\vx}_{\mathrm{b},0} \in \RN^{n_{\mathrm{b}}}$ contains the initial values of \eqref{equ:OurEquation}.

The nilpotent matrix $\mN \in \RN^{n_{\mathrm{b}} \times n_{\mathrm{b}}}$ with index of nilpotency $m\leq n_{\mathrm{b}}$ has generally the following block-diagonal Jordan canonical form
\begin{equation}
\label{equ:Nk}
	\mN = 
	\left(\begin{array}{cccc}
	\mN_{m_1} & \mNull & \cdots & \mNull\\
	\mNull & \mN_{m_2} & \cdots & \mNull\\
	\mNull & \cdots & \ddots & \vdots\\
	\mNull & \cdots & \mNull & \mN_{m_\ell}
	\end{array}\right)
\end{equation}
where each diagonal block $\mN_{m_k} \in \RN^{m_{k}\times m_{k}}{}$ has the form
\begin{equation}
\label{equ:Nmk}
	\mN_{m_k} = 
	\left(\begin{array}{ccccc}
	0 & 1 & 0 & \cdots & 0\\
	0 & 0 & 1 & \cdots & 0\\
	\vdots & \vdots & \ddots & \ddots & \vdots\\
	0 & \cdots & 0 & 0 & 1\\
	0 & \cdots & 0 & 0 & 0
	\end{array}\right)\,.
\end{equation}
Each $\mN_{m_k}$ has index of nilpotency $m_{k}$ and $\sum^{\ell}_{k=1} m_{k} = n_{\mathrm{b}}$.
Moreover, the index of nilpotency of $\mN$ in \eqref{equ:Nk} is given by
\begin{equation*}
	m = \max_{k = 1,2,\dots,\ell} m_{k}\;.
\end{equation*}
Thus \eqref{equ:OurEquation} decomposes into $\ell$ equations of the form
\begin{subequations}
\label{equ:SolutionElementaryN}
\begin{align}
\label{equ:SolutionElementaryN_a}
	\mN_{m}\, \vx'(t) &= \vx(t) + \vu(t)\,, \qquad t\in\TT\\
\label{equ:SolutionElementaryN_b}	
	\text{with}\quad \vx(0) &= \vx_{0}\,,
\end{align}
\end{subequations}
where $\mN_{m} \in \RN^{m \times m}$ is a nilpotent matrix of the form \eqref{equ:Nmk} with index of nilpotency $m$.
The simple structure $\mN_{m}$ makes it easy to solve \eqref{equ:SolutionElementaryN}  explicitly:

\begin{lemma}
\label{lem:SolutionElementaryEqu}
For an arbitrary $m\in\NN$, consider \eqref{equ:SolutionElementaryN_a} on $\TT = [0,T] \subset \RN$ with $T\in\RNc$.
The unique solution of \eqref{equ:SolutionElementaryN_a} is given by
\begin{equation}
\label{equ:Solution}
	\vx(t) = -\textstyle\sum^{m-1}_{k=0} \mN^{k}_{m}\, \vu^{(k)}(t)\,, \quad t\in\TT\,.
\end{equation}
\end{lemma}
\begin{remark}
Solution~\eqref{equ:Solution} shows in particular that for \eqref{equ:SolutionElementaryN_a} to be solvable, $\vu$ needs to be sufficiently differentiable, namely the $k$-th entry of $\vu$ needs to be $k-1$-times differentiable.
\end{remark}

Note that \eqref{equ:Solution} is the unique solution of \eqref{equ:SolutionElementaryN_a}, without requiring any initial conditions.
Consequently, if one considers the initial value problem \eqref{equ:SolutionElementaryN} with a fixed $\vu\in\C^{(m-1)}(\TT,\RN^{m})$, then there exists only one admissible vector of initial values, namely 
\begin{equation}
\label{equ:InitialCondN}
	\vx_{0}
	= -\textstyle\sum^{m-1}_{k=0} \mN^{k}_{m}\, \vu^{(k)}(0)\,.
\end{equation}
Conversely, for fixed initial values $\vx_{0} \in \RN^{m}$, there exist several functions $\vu \in \C^{(m-1)}(\TT,\RN^{m})$ that are compatible with the initial values $\vx_{0}$ in the sense that \eqref{equ:InitialCondN} is satisfied. In this case, we say that $\vu$ is an \emph{admissible} input.

Because of the block-diagonal form of the coefficient matrix \eqref{equ:Nk}, the statement of Lemma~\ref{lem:SolutionElementaryEqu} can easily be extended to equations of the form \eqref{equ:OurEquation}.

\begin{lemma}
\label{lem:SolutionNEqu}
Consider Equation \eqref{equ:OurEquation_a} on $\TT = [0,T] \subset \RN$ with $T\in\RNc$ and with a coefficient matrix $\mN\in\RN^{\nb\times \nb}$ of the form \eqref{equ:Nk} with index of nilpotency $1 \leq m \leq \nb$.
The unique solution of \eqref{equ:OurEquation_a} is given by 
\begin{equation}
\label{equ:SolutionE}
	\widetilde{\vx}_{\mathrm{b}}(t) = -\textstyle\sum^{m-1}_{k=0} \mN^{k} \widetilde{\vu}^{(k)}_{\mathrm{b}}(t)\,, 
	\qquad t\in\TT\,.
\end{equation}
\end{lemma}

\begin{remark}
Similarly as before, \eqref{equ:SolutionE} shows that the input $\widetilde{\vu}_{\mathrm{b}}$ needs to be sufficiently differentiable for \eqref{equ:OurEquation} to be solvable.
Moreover, for a fixed input $\widetilde{\vu}_{\mathrm{b}} \in \C^{m-1}(\TT,\RN^{\nb})$, there exists only one vector of initial values $\widetilde{\vx}_{\mathrm{b},0}\in\RN^{\nb}$ that is compatible with $\widetilde{\vu}_{\mathrm{b}}$.
Conversely, fixing the initial values $\widetilde{\vx}_{\mathrm{b},0}\in\RN^{\nb}$ there exists several admissible inputs $\widetilde{\vu}_{\mathrm{b}} \in \C^{m-1}(\TT,\RN^{\nb})$. 
\end{remark}

Since the index of nilpotency $m\in\NN$ of the matrix $\mN$ in the previous lemma is equal to the index of the LDAE \eqref{equ:LDAE},
Lemma~\ref{lem:SolutionNEqu} shows in particular that the index of an LDEA with regular coefficient matrix pair $(\mE,\mA)$ provides a necessary condition on the smoothness of $\vu$.
Thus if $m = \ind(\mE,\mA)$ is the index of an LDAE as in \eqref{equ:LDAE}, then $\vu$ in \eqref{equ:LDAE} needs to be at least $(m-1)$-times differentiable.

Moreover, the solution given by Lemma~\ref{lem:SolutionNEqu} illustrates that equations that are associated with the eigenvalues at infinity of $(\mE,\mA)$ describe a purely algebraic relation between $\widetilde{\vx}_{\mathrm{b}}(t)$ and $\widetilde{\vu}_{\mathrm{b}}(t)$ (including its derivatives up to order $m-1$).

\begin{figure}
\begin{center}
\begin{tikzpicture}[scale=0.8]	
	\draw[thick] (0,4) -- (1.2,4);
	\draw[thick] (2.5,4) -- (3.8,4);
	\draw[thick] (0,2) -- (1.775,2);
	\draw[thick] (1.925,2) -- (3.8,2);
	\draw[thick] (-0.1,4) circle (0.1);
	\draw[thick] (-0.1,2) circle (0.1);
	\draw[->,blue] (-0.1,3.8) -- (-0.1,2.2);
	\draw[blue] (-0.8,3)  node {$u_{0}(t)$};
	\draw[->,blue, bend angle = 45, bend left] (1.0,4.5) to (2.7,4.5);
	\draw[blue] (1.85,5.2)  node {$u_{\mathrm{L}}(t)$};
	\draw[->,blue, bend angle = 45, bend left] (4.2,3.85) to (4.2,2.15);
	\draw[blue] (5.2,3.0)  node {$u_{\mathrm{R}}(t)$};
	\draw[->,blue, bend angle = 45, bend left] (2.7,1.5) to (1.0,1.5);
	\draw[blue] (1.85,0.8)  node {$u_{\mathrm{C}}(t)$};
	\draw[thick] (3.8,2) -- (3.8,2.35);
	\draw[thick] (3.8,3.65) -- (3.8,4);
	\draw[very thick] (3.55,2.35) rectangle (4.05,3.65);
	\draw (3.2,3.0)  node {\bfseries R};
	\draw[very thick] (1.85-0.075,1.6) -- (1.85-0.075,2.4);
	\draw[very thick] (1.85+0.075,1.6) -- (1.85+0.075,2.4);	
	\draw (1.85,2.7)  node {\bfseries C};	
	\filldraw[very thick] (1.2,3.75) rectangle (2.5,4.25);
	\draw (1.85,3.5)  node {\bfseries L};
\end{tikzpicture}	
\end{center}
\vspace{-1ex}
\caption{A series connection of an inductor $L$, a resistor $R$, and a capacitor $C$ with input voltage $u_{0}$.}
\label{fig:LRC}
\end{figure}
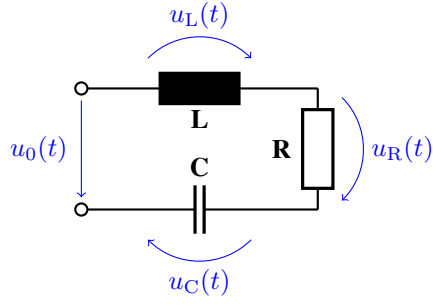

\subsection{A simple example}
\label{sec:ExampleKirchhoff}

To illustrate the previously introduced concepts, we consider the elementary example of a resonant circuit in Fig.~\ref{fig:LRC}.
It consists of a series connection of an inductor $L$, resistor $R$, and capacitor $C$. 
The circuit is excited by an external time-varying voltage $u_{0}$, and we are interested in the time-dependent behavior of the voltages $u_{L}$, $u_{R}$, and $u_{C}$ across the inductor, the resistor, and the capacitor, respectively.

Using Kirchhoff's circuit laws and the well known voltage-current relations on resistors, capacitors and inductors, one obtains the following three equations relating the voltages over $L$, $R$, and $C$ with the driving voltage $u_{0}$:
\begin{align*}
	u_{\mathrm{L}}(t) + u_{\mathrm{R}}(t) + u_{\mathrm{C}}(t) &= u_{0}(t)\\
	\tfrac{1}{R} u'_{\mathrm{R}}(t) &= \tfrac{1}{L} u_{\mathrm{L}}(t)\\
	C\, u'_{\mathrm{C}}(t) &= \tfrac{1}{R} u_{\mathrm{R}}(t)\;.
\end{align*}
To simplify the notation, we assume in this example that $RC = 1$ and $L/R = 1$. 
Now we write the previous three equations in the form of an LDAE:
\begin{equation*}
	\begin{pmatrix}
	0 & 0 & 0\\
	0 & 1 & 0\\
	0 & 0 & 1
	\end{pmatrix}\!
	\begin{pmatrix}
		u'_{\mathrm{L}}(t) \\ u'_{\mathrm{R}}(t) \\ u'_{\mathrm{C}}(t)
	\end{pmatrix}
	\!= \!
	\begin{pmatrix}
	1 & 1 & 1\\
	1 & 0 & 0\\
	0 & 1 & 0
	\end{pmatrix}\!
	\begin{pmatrix}
	u_{\mathrm{L}}(t) \\ u_{\mathrm{R}}(t) \\ u_{\mathrm{C}}(t)
	\end{pmatrix}
	+ 
	\begin{pmatrix}
	-u_{0}(t) \\ 0 \\ 0
	\end{pmatrix}\! .
\end{equation*}
If one defines the vectors $\vx(t) = (u_{\mathrm{L}}(t), u_{\mathrm{R}}(t), u_{\mathrm{C}}(t))^{\T}$ and
$\vu(t) = (-u_{0}(t), 0, 0)^{\T}$ and the coefficient matrices
\begin{equation*}
	\mE = \begin{pmatrix}
	0 & 0 & 0\\
	0 & 1 & 0\\
	0 & 0 & 1
	\end{pmatrix}
	\qquad\text{and}\qquad
	\mA = 	\begin{pmatrix}
	1 & 1 & 1\\
	1 & 0 & 0\\
	0 & 1 & 0
	\end{pmatrix}\,,
\end{equation*}
one obtains an equation of the form \eqref{equ:LDAE}.
The initial values are assumed to be $\vx(0) = \vzero$.

It is easy to verify that
\begin{equation*}
	\det(\lambda\mE - \mu\mA)
	= -\left(\lambda^{2} + \mu\lambda + \mu^{2}\right)\mu\,,
\end{equation*}
which shows that the matrix pair $(\mE,\mA)$ is regular and 
that it has two finite eigenvalues and one eigenvalue at infinity.
Its index is $\ind(\mE,\mA) = 1$.
Now it is an easy exercise to transform $(\mE,\mA)$ into the canonical from \eqref{equ:MazrixWeierstrassCF}.
Applying the transformation \eqref{equ:CoordinateTranfor}, with the transformation matrices
\begin{equation*}
	\mQ = \begin{pmatrix}
	-1 & -1 & 1\\
	1 & 0 & 0\\
	0 & 1 & 0\\
	\end{pmatrix}
	\qquad\text{and}\qquad
	\mP = \begin{pmatrix}
	0 & 0 & 1\\
	1 & 0 & 1\\
	0 & 1 & 0\\
	\end{pmatrix}	
\end{equation*}
one obtains the block diagonal form
\begin{equation*}
	\begin{pmatrix}
	1 & 0 & 0\\
	0 & 1 & 0\\
	0 & 0 & 0
	\end{pmatrix}
	\widetilde{\vx}'(t)
	=
	\begin{pmatrix}
	-1 & -1 & 0\\
	1 & 0 & 0\\
	0 & 0 & 1
	\end{pmatrix}
	\widetilde{\vx}(t)
	+ 
	\widetilde{\vu}(t)\,,
\end{equation*}
with initial values $\widetilde{\vx}(0) = 0$.
The upper diagonal block corresponds to the two finite eigenvalues and is an ODE with constant coefficients for $\widetilde{x}_{1}$ and $\widetilde{x}_{2}$.
The lower diagonal block corresponds to the eigenvalue at infinity and it is a purely algebraic equation which relates $\widetilde{x}_{3}$ and $\widetilde{u}_{3}$.

\section{Non-computable Solutions of LDAEs}
\label{sec:MainResults}

We investigate now the computability of solutions of singular LDAEs of the form \eqref{equ:LDAE}.
Thereby, we always assume that the matrix pair $(\mE,\mA)$ is regular and that both matrices and the initial values $\vx_{0}$ are computable. 
In particular, we always assume that the matrices have size $n\times n$, that $\mE$ is singular, and that $\ind(\mE,\mA) = m > 0$.
Then, for an arbitrary $\vu$, the initial value problem \eqref{equ:LDAE} is not necessarily solvable, but $\vu$ has to be compatible with the initial values $\vx_{0}$.
If $\vu$ and $\vx_{0}$ are compatible then \eqref{equ:LDAE} has a unique solution.
In other words, for fixed initial values $\vx_{0}$ there exists a whole set $\mathcal{U}(\vx_{0}) \subset \C(\TT,\RN^{n})$ of inputs $\vu$ for which \eqref{equ:LDAE} has a unique solution.
Any such $\vu \in \mathcal{U}(\vx_{0})$ is said to be \emph{admissible} for the initial value problem \eqref{equ:LDAE}.

We are going to show that for all computable initial values $\vx_{0}\in\RNc^{n}$ there always exists an admissible $\vu\in\mathcal{U}(\vx_{0})$ such that the unique solution $\vx$ of \eqref{equ:LDAE} is \emph{not} computable.
In fact we even show that there exist points $\tau\in\TT\cap\RNc$ for which $\vx(\tau) \notin \RNc^{n}$.
It will be shown that this occurs even for functions $\vu\in\mathcal{U}(\vx_{0})$ with strong regularity properties, namely for $\vu \in \C^{m}(\TT,\CN^{n})$ for which all derivatives up to order $m-2$ are computable continuous functions, and where $\vu^{(m-1)}$ is an absolute continuous, computable function in the Wiener algebra $\W(\TT,\RN^{n})$ and in
$\H12c(\TT,\CN^{n})$ .

The proofs of our main results are based on particular functions that are given in the following lemma.

\begin{lemma}
\label{lem:MainLemma}
Let $m\in\NN$, $m\geq 1$, and $\delta \in \RNc$, $\delta > 0$ be arbitrary and let $\alpha_{*} \in \Cweak$ be an arbitrary weakly computable real number.
There exists a computable continuous function $u \in \Cc(\TT)$ on $\TT = [0,\delta] \subset \RN$ such that
\begin{enumerate}
\item $u \in \C^{m-1}(\TT)$
\item $u^{(\ell)} \in \Cc(\TT)$ for all $0 \leq \ell \leq (m-2)$
\item $u^{(m-1)}\in\A(\TT)$ is absolute continuous
\item $u^{(m-1)} \in \H12c(\TT)$
\item $u^{(m-1)}(0) \in \RNc$ and $u^{(m-1)}(\frac{\delta}{2}) = \alpha_{*}$.
\end{enumerate}
\end{lemma}
\begin{remark}
The long and technical proof of this statement can be found in the separate Section~\ref{sec:Proof_AuxResultsAndMainLemma} at the end of this paper.
\end{remark}
For us, the following conclusion of Lemma~\ref{lem:MainLemma} will be of importance:
Since $\Cweak$ is a proper subset of $\RNc$, we can always choose an $\alpha_{*} \in \Cweak$ that is not a computable number.
Then Lemma~\ref{lem:MainLemma} shows that on any arbitrary computable interval $\TT = [0,\delta]\subset\RN$ there exists a point $\tau = \delta/2\in\RNc$ such that $u^{(m-1)}(\tau) \notin \RNc$.

\subsection{Solutions of the elementary nilpotent equation}

We start by considering elementary equations of the form \eqref{equ:SolutionElementaryN}.
As mentioned before (cf. discussion following Lemma~\ref{lem:SolutionElementaryEqu}), the initial values $\vx_{0}\in\RN^{m}$ and the input $\vu\in\Cc(\TT,\RN^{m})$ need to be compatible in order that \eqref{equ:SolutionElementaryN} is solvable.
Thus, fixing the initial values $\vx_{0} \in \RN^{(m)}$ implies restrictions on $\vu \in \Cc(\TT,\RN^{m})$ for \eqref{equ:SolutionElementaryN} to be solvable.
The following auxiliary lemma shows how, for fixed initial values $\vx_{0}$, one can construct an admissible $\vu$ (with strong regularity properties) .
This result is needed later to prove our main results in this section.

\begin{lemma}
\label{equ:Construct_u}
Fix arbitrary numbers $m\in\NN$, $\delta\in\RNc$ with $\delta > 0$,
and consider \eqref{equ:SolutionElementaryN} on the interval $\TT = [0,\delta]$ with arbitrary computable initial values
\begin{equation*}
	\vx_{0} = (x_{0,1}, x_{0,2}, \dots, x_{0,m})^{\T} \in \RNc^{m}\,.
\end{equation*}
There exists an admissible $\vu\in\Cc(\TT,\RN^{m}) \cap \C^{\infty}(\TT,\RN^{m})$ with
\begin{itemize}
\item $\vu^{(\ell)} \in \Cc(\TT,\RN^{m})$
\item $\vu^{(\ell)}$ is absolute continuous
\item $\vu^{(\ell)} \in \W(\TT,\RN^{m})$
\item $\vu^{(\ell)} \in \H12c(\TT,\RN^{m})$
\end{itemize}
for every $\ell = 0,1,2,\dots$, and such that
the unique solution $\vx$ of \eqref{equ:SolutionElementaryN} belongs to $\Cc(\TT,\RN^{m})\cap\C^{\infty}(\TT,\RN^{m})$ with
$\vx^{(\ell)} \in \Cc(\TT,\RN^{m})$ for all $\ell = 1,2,\dots$.
\end{lemma}

\begin{IEEEproof}
Choose $d\in\NN$ such that $2d+1 \geq m$ and let $p \in \Cc(\TT)$ be a trigonometric polynomial of degree $d\in\NN$ 
with coefficients $\left\{ a_{k} \right\}^{d}_{k=0}\subset\RNc$ and $\left\{ b_{k} \right\}^{d}_{k=1}\subset\RNc$,
i.e.
\begin{equation}
\label{equ:TrigPolyProof}
	p(t) = \textstyle\frac{a_{0}}{2} + \sum^{d}_{k=1}  \big[ a_{k} \cos\left(\tfrac{2\pi}{\delta}kt\right) + b_{k} \sin\left(\tfrac{2\pi}{\delta}kt\right) \big]\,,	
\end{equation}
with $t\in\TT$ and such that
\begin{equation}
\label{equ:CondTrigPoly1}
	p^{(m-\ell)}(0) = - x_{0,\ell},
	\quad \ell = 1,2,\dots,m\,.
\end{equation} 
Since $2d+1 \geq m$, such a polynomial $p$ always exists and one can use the $m$ linear equations \eqref{equ:CondTrigPoly1} to determine the coefficients of \eqref{equ:TrigPolyProof} (which are not unique if $2d+1 > m$).
Writing the equations \eqref{equ:CondTrigPoly1} in details, one sees that the corresponding matrix has a Vandermonde-like structure and that all of its entries are computable numbers, because $\delta \in \RNc$.
Therefore \eqref{equ:CondTrigPoly1} has always at least one solution, and since $\RNc$ forms a real closed field (in particular, $\RNc$ is closed with respect to the basic arithmetic operations) the solution consists always of computable numbers.
Consequently, the coefficients of the polynomial \eqref{equ:TrigPolyProof} that satisfies \eqref{equ:CondTrigPoly1} always has computable coefficients. 

Finally, we define $\vu \in \Cc^{\infty}(\TT,\RN^{m})$, by
\begin{equation*}
	\vu(t) = \big(0,0,\dots,0, p(t) \big)^{\T}\,,
	\qquad t\in\TT\,.
\end{equation*}
Then Lemma~\ref{lem:SolutionElementaryEqu} shows that the solution of \eqref{equ:SolutionElementaryN} with this particular $\vu$ is given by
\begin{equation*}
	\vx(t) = -\big(p^{(m-1)}(t),p^{(m-2)}(t),\dots,p'(t), p(t) \big)^{\T}\,,
	\quad t\in\TT\,,
\end{equation*}
and because of \eqref{equ:CondTrigPoly1}, this solution satisfies the initial condition $\vx(0) = \vx_{0}$.
The other properties of $\vu$ and the statements that $\vx\in\Cc(\TT,\RN^{m})\cap\C^{\infty}(\TT,\RN^{m})$
with $\vx^{(\ell)}\in\Cc(\TT,\RN^{m})$
are apparent because $p$ is a computable trigonometric polynomial.
\end{IEEEproof}

The previous lemma showed that for arbitrary computable initial values $\vx_{0} \in \RNc^{m}$, one can always construct a computable input $\vu\in\Cc^{\infty}(\TT,\RN^{m})$ satisfying strong regularity properties such that the corresponding solution $\vx$ of \eqref{equ:SolutionElementaryN} is a computable continuous function.

The next lemma shows that for arbitrary computable initial values $\vx_{0} \in \RNc^{m}$, there always exists a computable input $\vu\in\Cc(\TT,\RN^{m})$ which has almost similar strong regularity properties as the input function in Lemma~\ref{equ:Construct_u} but such that the solution of \eqref{equ:SolutionElementaryN} is not computable.

\begin{lemma}
\label{lem:NotCompSolutionN}
Fix arbitrary numbers $m\in\NN$, $\delta\in\RNc$ with $\delta > 0$,
and consider \eqref{equ:SolutionElementaryN} on the interval $\TT = [0,\delta]$ with arbitrary computable initial values
$$\vx_{0} = (x_{0,1}, x_{0,2}, \dots, x_{0,m})^{\T} \in \RNc^{m}\,.$$
There exists an admissible $\vu \in \Cc(\TT,\RN^{m}) \cap \C^{m-1}(\TT,\RN^{m})$ with 
\begin{itemize}
\item $\vu^{(m-1)}$ is absolute continuous
\item $\vu^{(m-1)} \in \W(\TT,\RN^{m})$
\item $\vu^{(m-1)} \in \H12c(\TT,\RN^{m})$ 
\end{itemize}
such that the unique solution $\vx\in\C(\TT,\RN^{m})$ of \eqref{equ:SolutionElementaryN} 
is not a computable function on $\TT$, and in particular such that
\begin{equation}
\label{equ:xATdelta2}
	\vx(\delta/2)
	= \big( x_{1}(\delta/2), \dots, x_{m}(\delta/2) \big)^{\T} \notin \RNc^{m}\,.
\end{equation}
\end{lemma}

\begin{IEEEproof}
After fixing $m\in\NN$ and $\delta\in\RNc$, we choose an arbitrary $\alpha_{*} \in \Cweak\backslash\RNc$.
Based on these numbers, let $u \in \Cc(\TT)$ be the corresponding function as given in Lemma~\ref{lem:MainLemma}.
Therewith, we define $\vu_{1}(t) = \big(0,0,\dots,0,u(t)\big)^{\T}$ and consider the corresponding equation \eqref{equ:SolutionElementaryN_a},
i.e. $\mN_{m}\, \vx'_{1}(t) = \vx_{1}(t) + \vu_{1}(t)$.
According to Lemma~\ref{lem:SolutionElementaryEqu}, this equation has a unique solution 
\begin{equation}
\label{equ:Solution_x1}
	\vx_{1}(t) = -\left( u^{(m-1)}(t), u^{(m-2)}(t), \dots, u'(t), u(t) \right)^{\T},
\end{equation}
for $t\in\TT$.
Therewith, we define $\widehat{\vx}_{0} = \vx_{0} - \vx_{1}(0)$ and consider the initial value problem 
\begin{subequations}
\begin{align*}
	\mN_{m}\, \vx'_{2}(t) &= \vx_{2}(t) + \vu_{2}(t)\,, \qquad t\in\TT\,,\\
	\text{with}\quad \vx_{2}(0) &= \widehat{\vx}_{0}\,.
\end{align*}
\end{subequations}
For fixed initial values $\widehat{\vx}_{0}$, we construct an admissible input $\vu_{2}\in\Cc(\TT,\RN^{m}) \cap \C^{\infty}(\TT,\RN^{m})$ using Lemma~\ref{equ:Construct_u}.

Therewith, we define $\vu = \vu_{1} + \vu_{2}$ and consider \eqref{equ:SolutionElementaryN}.
The defined $\vu$ is admissible for the given initial values $\vx_{0}$. 
Indeed, by the linearity of \eqref{equ:SolutionElementaryN_a}, we have $\vx(t) = \vx_{1}(t) + \vx_{2}(t)$, $t\in\TT$ and in particular
$\vx(0) = \vx_{1}(0) + \vx_{2}(0) = \vx_{1}(0) + \widehat{\vx}_{0} = \vx_{0}$.

%
The properties of $\vu_{1}$ are determined by the properties of the function $u$ in its last component. These properties are listed in Lemma~\ref{lem:MainLemma}.
In particular, Properties 1) -- 4) of Lemma~\ref{lem:MainLemma}, show immediately that $\vu_{1} \in \Cc(\TT,\RN^{m}) \cap \C^{m-1}(\TT,\RN^{m})$ and that $\vu_{1}^{(m-1)}$ has the three properties claimed in the statement of Lemma~\ref{lem:NotCompSolutionN}. By Lemma~\ref{equ:Construct_u}, $\vu_{2}\in\Cc(\TT,\RN^{m}) \cap \C^{\infty}(\TT,\RN^{m})$ also posses all theses properties.
This shows that also $\vu = \vu_{1} + \vu_{2} \in\Cc(\TT,\RN^{m}) \cap \C^{m-1}(\TT,\RN^{m})$ has these three properties.

Property~5) of Lemma~\ref{lem:MainLemma} shows that $u^{(m-1)}(\delta/2) \notin\RNc$. 
This implies that the first component of $\vx_{1}$, given in \eqref{equ:Solution_x1}, is not computable at $\delta/2$.
Consequently $\vx_{1}(\delta/2) \notin\RNc^{m}$.
However, Lemma~\ref{equ:Construct_u} shows that $\vx_{2}$ is a computable function. Consequently $\vx_{2}(\delta/2) \in \RNc^{m}$ and so we have
$\vx(\delta/2) = \vx_{1}(\delta/2) + \vx_{2}(\delta/2) \notin\RNc^{m}$.
\end{IEEEproof}

\subsection{Solution of general singular LDAEs}

Equation \eqref{equ:SolutionElementaryN} is only one diagonal block of the more general equation \eqref{equ:OurEquation}.
Since all diagonal blocks of \eqref{equ:OurEquation} have the form \eqref{equ:SolutionElementaryN}, but with potentially different sizes $m_{1}, \dots, m_{\ell}$,  
it is straight forward to generalize the statement of Lemma~\ref{lem:NotCompSolutionN} to equations of the form \eqref{equ:OurEquation}:
\begin{lemma}
\label{lem:N}
Fix an arbitrary $\delta\in\RNc$ with $\delta > 0$ and
consider \eqref{equ:OurEquation} with a nilpotent, block-diagonal coefficient matrix $\mN \in \CN^{\nb\times \nb}$ of the form \eqref{equ:Nk}
with index of nilpotency $m\leq \nb$ and with arbitrary initial values $\widetilde{\vx}_{\mathrm{b},0} \in \RNc^{\nb}$, on the interval $\TT = [0,\delta]$.\\
There exists an admissible $\widetilde{\vu}_{\mathrm{b}} \in \Cc(\TT,\RN^{\nb}) \cap \C^{m-1}(\TT,\RN^{\nb})$ with 
\begin{itemize}
\item $\widetilde{\vu}^{(m-1)}_{\mathrm{b}}$ is absolutely continuous
\item $\widetilde{\vu}^{(m-1)}_{\mathrm{b}} \in \W(\TT,\RN^{\nb})$
\item $\widetilde{\vu}^{(m-1)}_{\mathrm{b}} \in \H12c(\TT,\RN^{\nb})$
\end{itemize}
such that the unique solution $\widetilde{\vx}_{\mathrm{b}}\in\C(\TT,\RN^{\nb})$ of \eqref{equ:OurEquation} is not a computable continuous function.
In particular, we have
\begin{equation}
\label{equ:xATdelta2_N}
	\widetilde{\vx}(\delta/2)
	= \big( \widetilde{x}_{1}(\delta/2), \dots, \widetilde{x}_{\nb}(\delta/2) \big)^{\T} \notin \RNc^{\nb}\,.
\end{equation}
\end{lemma}

\begin{IEEEproof}
Assume the coefficient matrix $\mN \in \CN^{\nb\times \nb}$ has block-diagonal form \eqref{equ:Nk} with $1 \leq \ell \leq \nb$ diagonal blocks, and let $m_{i} > 0$, $i \in \left\{ 1,\dots,\ell \right\}$ be the index of nilpotency of the $i$-th diagonal block.
Since $m = \max_{i\in\left\{1,\dots,\ell \right\}} m_{i}$, there exists at least one diagonal block of $\mN$ with index of nil-potency $m$.
Assume this is the $k$-th block, i.e. $m_k = m$.

Thus the whole initial value problem \eqref{equ:OurEquation} splits into $\ell$ independent initial value problems of the form \eqref{equ:SolutionElementaryN}.
For the $k$-th  initial value problem, we construct an admissible $\widetilde{\vu}_{\mathrm{b},k} \in \Cc(\TT,\RN^{m}) \cap \C^{m-1}(\TT,\RN^{m})$ according to Lemma~\ref{lem:NotCompSolutionN}.
For all other initial value problem, we construct an admissible $\widetilde{\vu}_{\mathrm{b},i}\in\Cc(\TT,\RN^{m_i}) \cap \C^{\infty}(\TT,\RN^{m_i})$ according to Lemma~\ref{equ:Construct_u}.
Therewith, we define
\begin{equation*}
	\widetilde{\vu}_{\mathrm{b}}(t)
	= \left(\widetilde{\vu}_{\mathrm{b},1}(t), \dots, \widetilde{\vu}_{\mathrm{b},\ell}(t)\right)^{\T},
	\quad t\in\TT\;.
\end{equation*}
By Lemmas~\ref{equ:Construct_u} and \ref{lem:NotCompSolutionN}, all subblocks $\widetilde{\vu}_{\mathrm{b},i}$, $i = 1,2,\dots\ell$ have the required properties. Consequently, also $\widetilde{\vu}_{\mathrm{b}}$ has these properties.

It is clear, that the corresponding solution of \eqref{equ:OurEquation} is given by
\begin{equation*}
	\widetilde{\vx}_{\mathrm{b}}(t)
	= \left(\widetilde{\vx}_{\mathrm{b},1}(t), \dots, \widetilde{\vx}_{\mathrm{b},\ell}(t)\right)^{\T},
	\quad t\in\TT\;.
\end{equation*}
where $\widetilde{\vx}_{\mathrm{b},i}$ is the unique solution of the $i$-th initial value problem of the form \eqref{equ:SolutionElementaryN}.
From Lemma~\ref{equ:Construct_u}, we know that $\widetilde{\vx}_{\mathrm{b},i}(\delta/2) \in \RNc^{m_{1}}$ for all $i\neq k$, 
whereas Lemma~\ref{lem:NotCompSolutionN} shows that $\widetilde{\vx}_{\mathrm{b},k}(\delta/2) \notin \RNc^{m_{k}}$.
This shows finally that $\widetilde{\vx}_{\mathrm{b}}(\delta/2) \notin \RNc^{\nb}$.
\end{IEEEproof}

Finally we show that the unique solution of singular LDAEs with regular matrix pairs $(\mE,\mA)$ are generally not computable.

\begin{theorem}
\label{thm:MainTheorem}
Let \eqref{equ:LDAE} be a singular LDAE with a regular pair $(\mE,\mA)$ of computable coefficient matrices of size $n\times n$, with $\ind(\mE,\mA) = m$,
and with initial value $\vx_{0}\in\RNc^{n}$.
Let $\delta\in\RNc$, $\delta > 0$ be arbitrary and let $\TT = [0,\delta] \subset \RN$.
Then there always exist an admissible $\vu\in\Cc(\TT,\RN^{n}) \cap \C^{m-1}(\TT,\RN^{n})$ such that the unique solutions $\vx$ of \eqref{equ:LDAE} satisfy
\begin{equation*}
	\vx(\delta/2) \notin\RNc^{n}\,.
\end{equation*}
\end{theorem}

\begin{IEEEproof}
Consider \eqref{equ:LDAE} in its Weierstrass canonical form \eqref{equ:WeierstrassCanonForm}.
Since the LDAE is assumed to be singular, $(\mE,\mA)$ has at least one eigenvalue at infinity. 
Therefore the second line in \eqref{equ:WeierstrassCanonForm} is non-zero with a nilpotent matrix $\mN \in \RN^{n_{\mathrm{b}}\times n_{\mathrm{b}}}$ with index of nilpotency $m$.

Now, we choose for $\widetilde{\vu}_{\mathrm{b}} \in \Cc(\TT,\RN^{n_\mathrm{b}}) \cap \C^{m-1}(\TT,\RN^{n_{\mathrm{b}}})$ to be the vector constructed in the proof of Lemma~\ref{lem:N}.
For $\widetilde{\vu}_{\mathrm{a}}$, we can take any function from $\Cc(\TT,\RN^{\na})$ for the ODE in the first line of \eqref{equ:WeierstrassCanonForm} to have a unique solution $\widetilde{\vx}_{\mathrm{a}} \in \Cc(\TT,\RN^{\na})$.
Thus we choose in particular $\widetilde{\vu}_{\mathrm{a}} \in \Cc(\TT,\RN^{\na}) \cap \C^{m-1}(\TT,\RN^{\na})$.

However, according to Lemma~\ref{lem:N}, the solution $\widetilde{\vx}_{\mathrm{b}}$ of the second line of \eqref{equ:WeierstrassCanonForm} has the property that $\widetilde{\vx}_{\mathrm{b}}(\delta/2) \notin \RNc^{n_{\mathrm{b}}}$.
Since the matrices $\mQ$ and $\mP$ in \eqref{equ:TrafoToCanonicalForm}, that transform \eqref{equ:LDAE} into the canonical form \eqref{equ:WeierstrassCanonForm}, are computable, it follows that the solution of \eqref{equ:LDAE} always satisfies
\begin{equation*}
	\vx(\delta/2)
	= \mQ \left[\begin{array}{c}
	\vx_{\mathrm{a}}(\delta/2) \\ \vx_{\mathrm{b}}(\delta/2)
	\end{array}\right] \notin \RNc^{n}\,,
\end{equation*}
independent on the solution $\widetilde{\vx}_{\mathrm{a}}$, i.e. independent on the chosen $\widetilde{\vu}_{\mathrm{a}}$ for the first line in \eqref{equ:WeierstrassCanonForm}.
\end{IEEEproof}

\subsection{Discussion}

Let $(\mE,\mA)$ be a regular pair of computable matrices of size $n\times n$ with a singular $\mE$,  and let $\vx_{0} \in \RN^{n}$ be an arbitrary vector of computable initial values.
Then $(\mE,\mA)$ and $\vx_{0}$ define a singular LDAE \eqref{equ:LDAE} with a certain $\vu\in\C(\TT,\RN^{n})$.
Equation \eqref{equ:LDAE} will not be solvable for any arbitrary $\vu\in\C(\TT,\RN^{n})$, but $\vu$ has to be compatible with the initial values $\vx_{0}$ in order that \eqref{equ:LDAE} has a unique solution.
Let $\mathcal{U}(\vx_{0}) \subset \C(\TT,\RN^{n})$ be the set of all those functions for which  \eqref{equ:LDAE} has a unique solution.
Then Theorem~\ref{thm:MainTheorem} shows that the set $\mathcal{U}(\vx_{0})$ always contains functions $\vu$ such that the unique solution $\vx$ of \eqref{equ:LDAE} is not computable at some points on any 
arbitrary interval $[0,\delta] \subset\RN$.
Since \eqref{equ:LDAE} is linear, it is clear that this happens also for function $\vu$ with arbitrary absolute value of its small entries.

Theorem~\ref{thm:MainTheorem} shows in particular, that the non-computability of solutions of LDAEs occurs on arbitrary intervals $[0,\delta]\subset\RN$ with arbitrary small $\delta > 0$.
Therefore, our results are also of interest for non-linear differential-algebraic equations (DAEs) which frequently arise in robotics, molecular chemistry, and many other areas.
Because these non-linear DAEs are often linearized about a nominal operating state, assuming that the variation of the system states remain small about this nominal state. 
Then in a small neighborhood of the nominal state, the behavior of the system can be described approximately by a linear singular DAE. 
A typical example arises in control:  Assume one might to try to stabilize a certain system (that is described by a non-linear DAE) at a certain nominal state.
Then for small deviations from this nominal state, the dynamics of the system can be approximated by a singular LDAE.
Based on this LDAE, one can try to design an appropriate controller that keeps the system in the nominal state. (see, e.g., \cite{MillsGoldenberg_TransRA89}).
Since the non-computability of solutions of LDAEs as stated in Theorem~\ref{thm:MainTheorem} holds for arbitrary small intervals $[0,\delta]\subset\RN$ it shows that the solution of the LDAE can generally not be used to compute an approximative solution of non-linear singular DAE in a neighborhood of the nominal state. 
In fact, even the solution at nominal state itself is generally not computable.

This observation is also of interest, because non-linear DAEs can generally not be solved in closed form and there exists no general applicable numerical methods to solve such equations.
Therefore, one often approximates the non-linear DAE locally by LDAE and uses then standard numerical algorithms for LDAEs.

\section{Extensions and Further Discussions}

The linear differential-algebraic equation of the form \eqref{equ:LDAE} is only a special case of the more general differential-algebraic equation of the form
\begin{align}
\label{equ:DAE_F}
\begin{split}
	\mE \vx'(t) &= \mF(t,\vx(t))\,,
	\qquad t\in\TT\;,\\
	\text{with}\qquad
	\vx(0) &= \vx_{0}\;,
\end{split}
\end{align}
where $\mE \in \RN^{n\times n}$ is a coefficient matrix and $\mF : \TT\times\RN^{n} \to \RN^{n}$ is a time-dependent function that maps $\RN^{n}$ into $\RN^{n}$.
As before, $\vx_{0} \in \RN^{n}$ is the vector of initial values.
For better comparison with LDAEs \eqref{equ:LDAE} with a regular matrix pair $(\mE,\mA)$, we consider here only the case where $\mE$ is a square matrix.
As before, we always assume that $\mE$ and $\vx_{0}$ is a computable matrix and computable vector, respectively.

It is clear, that \eqref{equ:LDAE} is a special case of \eqref{equ:DAE_F} in which $\mF : \TT\times\RN^{n} \to \RN^{n}$ is given by
\begin{equation*}
	\mF(t,\vx) = \mA\, \vx + \vu(t)\,,
	\qquad \text{for}\ t\in\TT\ \text{and}\ \vx\in\RN^{n}\,,
\end{equation*}
where $\mA\in\RN^{n\times n}$ is a coefficient matrix and $\vu : \TT \to \RN^{n}$ a vector valued function on $\TT\subset\RN$.
Therefore Theorem~\ref{thm:MainTheorem} immediately implies the following observation for differential-algebraic equations of the form \eqref{equ:DAE_F}.

\begin{corollary}
Let $n\in\NN$ with $n\geq 3$, let $\delta\in\RNc$ with $\delta > 0$ and $\TT = [0,\delta]$.
Let $\mE\in\RN^{n\times n}$ be a singular, computable matrix and let $\vx_{0}\in\RN^{n}$ be an arbitrary computable vector.
There always exists an $\mF : \TT\times \RN^{n} \to \RN^{n}$ with
\begin{itemize}
\item $\mF(\cdot,\vx) \in \Cc(\TT,\RN^{n}) \cap \C^{(m-1)}(\TT,\RN^{n})$ for every $\vx\in\RNc^{n}$
\item ...
\end{itemize} 
for some $0 < m < n$ and
such that \eqref{equ:DAE_F} has a unique solution $\vx\in\C(\TT,\RN^{n})$ which is not a computable continuous function with the property
\begin{equation*}
	\vx(\delta/2) \notin\RNc^{n}\,.
\end{equation*}
\end{corollary}

\begin{IEEEproof}
We can choose $\mF$ as 
\begin{equation}
\label{equ:AffineRightSide}
	\mF(t,\vx) = \vx + \vu(t)\,,
	\qquad t\in\TT\,.
\end{equation}
Therewith \eqref{equ:DAE_F} becomes a LDAS \eqref{equ:LDAE} with matrix pair $(\mE,\mI)$.
Since $\det(\lambda\mE + I)$ is not the zero polynomial, $(\mE,\mI)$ is a regular matrix pair and because $\mE$ is assumed to be singular, we get a singular LDAE with index $m = \ind(\mE,\mI)$ that satisfies $0< m < n$.

Thus the resulting singular LDAE satisfies the conditions on Theorem~\ref{thm:MainTheorem}, and so we choose the vector valued function $\vu$ in \eqref{equ:AffineRightSide} as the $\vu\in\Cc(\TT,\RN^{n})\cap \C^{m-1}(\TT,\RN^{n})$ given by Theorem~\ref{thm:MainTheorem}.
Then by the properties of this $\vu$, it is clear that $\mF(\cdot,\vx)$ is a computable continuous functions and that it is $(m-1)$-times continuously differentiable for every $\vx\in\RNc^{n}$.
From the affine structure of \eqref{equ:AffineRightSide} it is immediately clear that 
\end{IEEEproof}

\section{Proof of Lemma~\ref{lem:MainLemma}}
\label{sec:Proof_AuxResultsAndMainLemma}

In this section, we will prove our central Lemma~\ref{lem:MainLemma}.
This proof will explicitly construct the function $u\in\Cc(\TT)$ whose properties are listed in Lemma~\ref{lem:MainLemma}.

The construction of $u$ is based on three main ingredients:
\begin{enumerate}
\item A particular computable sequence $\left\{ \alpha_{n} \right\}^{\infty}_{n=0}$ that is associated with a number $\alpha\in\Cweak$.
\item A particular set of trigonometric polynomials $p_{m,k}$.
\item A set $2\pi$-periodic functions $\left\{ w_{\ell} \right\}$ which is defined based on the sequence $\left\{ \alpha_{n} \right\}^{\infty}_{n=0}$ and the trigonometric polynomials $p_{m,k}$.
\end{enumerate}

In the attempt to make the construction in the proof of Lemma~\ref{lem:MainLemma} as clear as possible, we split the proof into several auxiliary statements which are given the following subsections:
Subsection~\ref{app:LemmaSequence} will define the sequence $\left\{ \alpha_{n} \right\}^{\infty}_{n=0}$ associated with a number $\alpha\in\Cweak$ and proves the necessary properties of this sequences.
Subsection~\ref{app:TrigPolynomials} defines the trigonometric polynomials $p_{m,k}$ and derives their main properties.
Subsection~\ref{app:Properties_wl} defines the functions $w_{\ell}$ and proves some properties of these functions.
After these preparations, Subsection~\ref{sec:ProofMainLemma} will finally present the proof of Lemma~\ref{lem:MainLemma}.

\subsection{A sequence associated with $\alpha\in\Cweak$}
\label{app:LemmaSequence}

Let $\alpha\in\Cweak$, $\alpha\neq 0$ be arbitrary. 
By the definition of $\Cweak$, there exists $\beta\in \Pu$ and $\gamma\in\Pu$ such that $\alpha = \beta - \gamma$.
Since $\beta,\gamma \in \Pu$, there exist two monotone non-decreasing computable sequences $\left\{ \beta_{n} \right\}_{n\in\NN}$ and $\left\{ \gamma_{n} \right\}_{n\in\NN}$ of computable numbers such that 
\begin{equation}
\label{equ:limitBetaAndGamma}
	\textstyle\lim_{n\to\infty} \beta_{n} = \beta
	\qquad\text{and}\qquad
	\lim_{n\to\infty} \gamma_{n} = \gamma\,.
\end{equation}
Therewith, we define the sequence $\left\{ \alpha_{n} \right\}^{\infty}_{n=0}$ by
\begin{equation}
\label{equ:DefSeqAlpha}
	\alpha_{0} = 0
	\qquad\text{and}\qquad
	\alpha_{n} = \beta_{n} - \gamma_{n}\,,\quad n = 1,2,\dots\;.
\end{equation}
The following lemma gives properties $\left\{ \alpha_{n} \right\}^{\infty}_{n=0}$, frequently needed in the subsequent proofs.

\begin{lemma}
\label{lem:SeqAlpha}
Let $\alpha\in\Cweak$ be arbitrary.
There exists a computable sequence $\left\{ \alpha_{n} \right\}^{\infty}_{n=0} \subset \RN_{c}$ with the following properties:
\begin{enumerate}
\item $\lim_{n\to\infty} \alpha_{n} = \alpha$.
\item $\lim_{R\to\infty} \sum^{\infty}_{n=R+1} \left|\alpha_{n} - \alpha_{n-1} \right| = 0$.
\item There exists a $Q_{\alpha} \in\NN$ such that
\begin{equation*}
	\textstyle\sum^{\infty}_{n=R+1} \left|\alpha_{n} - \alpha_{n-1}\right| \leq Q_{\alpha}\,,
	\qquad\text{for all}\ R\in\NN\,.
\end{equation*}
\end{enumerate}
\end{lemma}

\begin{IEEEproof}
Let $\left\{ \alpha_{n} \right\}^{\infty}_{n=0}$ be defined by \eqref{equ:DefSeqAlpha}.
From its definition, it is clear that $\left\{ \alpha_{n} \right\}^{\infty}_{n=0}$ is computable sequence of computable numbers with $\lim_{n\to\infty} \alpha_{n} = \lim_{n\to\infty} \beta_{n} - \lim_{n\to\infty} \gamma_{n} = \beta - \gamma = \alpha$.
To verify 2), we notice that for $R,N\in\NN$ with $N>R$, we have by the triangle inequality
\begin{multline*}
	\sum^{N}_{n=R+1} \left| \alpha_{n} - \alpha_{n-1}\right|
	\leq \sum^{N}_{n=R+1} \left| \beta_{n} - \beta_{n-1}\right| + \left| \gamma_{n} - \gamma_{n-1}\right|
	= \sum^{N}_{n=R+1} \left( \beta_{n} - \beta_{n-1}\right) + \left( \gamma_{n} - \gamma_{n-1}\right)
	= \beta_{N} - \beta_{R} + \gamma_{N} - \gamma_{R}\;,
\end{multline*}
using that $\left\{ \beta_{n} \right\}_{n\in\NN}$ and $\left\{ \gamma_{n} \right\}_{n\in\NN}$ are non-decreasing.
Because of \eqref{equ:limitBetaAndGamma}, this shows that
\begin{equation*}
	\textstyle\lim_{N,R\to\infty} \sum^{N}_{n=R+1} \left|\alpha_{n} - \alpha_{n-1} \right| = 0
\end{equation*}
and which proves in particular Statement 2).

Let $N\to\infty$ and using that $\left\{ \beta_{n} \right\}_{n\in\NN}$ and $\left\{ \gamma_{n} \right\}_{n\in\NN}$ are non-decreasing, we get
\begin{equation*}
	\textstyle\sum^{\infty}_{n=R+1} \left| \alpha_{n} - \alpha_{n-1}\right|
	= \beta + \gamma - \left(\beta_{R} + \gamma_{R}\right)
	\leq \beta + \gamma - \left(\beta_{1} + \gamma_{1}\right)
	\leq Q_{\alpha}
\end{equation*}
where $Q_{\alpha} \in \NN$ is defined to be the smallest integer that satisfies $\beta + \gamma - \left(\beta_{1} + \gamma_{1}\right) \leq Q_{\alpha}$.
\end{IEEEproof}

\begin{remark}
Note that the convergence in Statement 2) of Lemma~\ref{lem:SeqAlpha} is not necessarily effective.
\end{remark}

\subsection{A set of trigonometric polynomials}
\label{app:TrigPolynomials}

Let $\TT = [0,2\pi] \subset \RN$, and let $m\in\NN$, $m>1$ be an arbitrary but fixed positive integer.
Let $r\in\NN$ be the largest integer that satisfies $0 \leq 4r < m$.
Depended on $m$, we define for any $k\in\NN$, $k>1$, the following trigonometric polynomial
\begin{equation}
\label{equ:definition_pmk}
	p_{m,k}(t) = \left\{\begin{array}{lll}
	+ \frac{1}{C(k)} \sum^{k}_{n=1} \frac{\cos(nt)}{n^{m}} & \text{if} & m = 4r + 1\\[0.7ex]
	+ \frac{1}{C(k)} \sum^{k}_{n=1} \frac{\sin(nt)}{n^{m}} & \text{if} & m = 4r + 2\\[0.7ex]
	- \frac{1}{C(k)} \sum^{k}_{n=1} \frac{\cos(nt)}{n^{m}} & \text{if} & m = 4r + 3\\[0.7ex]
	- \frac{1}{C(k)} \sum^{k}_{n=1} \frac{\sin(nt)}{n^{m}} & \text{if} & m = 4r + 4
	\end{array}\right.
\end{equation}
%
for $t\in\TT$ and with the constant 
\begin{equation}
\label{equ:Bound_Ck}
	C(k)
	:= \sum^{k}_{n=1} \frac{1}{n}
	\geq \int^{k+1}_{1} \frac{\d\tau}{\tau}
	= \log(k+1)\,.
\end{equation} 
Note the definition \eqref{equ:definition_pmk} yields
\begin{equation}
\label{equ:deribative_pml_mMinus1}
	p^{(m-1)}_{m,k}(t) = + \frac{1}{C(k)} \sum^{k}_{n=1} \frac{\cos(nt)}{n}\,,\quad t\in\TT\,,
\end{equation}
for the $(m-1)$-th derivative of $p_{m,k}$, independent on the specific $m$.
All other derivatives $p^{(\ell)}_{m,k}$ with $0\leq \ell \leq m-2$ have the form
\begin{equation}
\label{equ:derivatives_pml}
	p^{(\ell)}_{m,k}(t) = \textstyle\pm\frac{1}{C(k)} \sum^{k}_{n=1} \frac{\cos(nt)}{n^{m-\ell}} \qquad\text{or}\qquad
	p^{(\ell)}_{m,k}(t) = \textstyle\pm\frac{1}{C(k)} \sum^{k}_{n=1} \frac{\sin(nt)}{n^{m-\ell}}
\end{equation}
depended on the concrete $m$ and $\ell$.

Next, we derive properties of these polynomials that are needed in the subsequent proof.
From the formulas \eqref{equ:derivatives_pml} for the derivatives, we immediately get for all $0\leq \ell \leq m-2$
\begin{equation}
\label{equ:Bound_pmkell}
	\left|p^{(\ell)}_{m,k}(t)\right|
	\leq \frac{1}{C(k)} \sum^{k}_{n=1} \frac{1}{n^{m-\ell}} 
	\leq \frac{1}{C(k)} \sum^{\infty}_{n=1} \frac{1}{n^{2}}
	\leq \frac{1}{\log(k+1)}\frac{\pi^{2}}{6}\;,
	\qquad\text{for all}\ t\in\TT\;.
\end{equation}
In the case $\ell = (m-1)$, given by \eqref{equ:deribative_pml_mMinus1}, we obtain
\begin{equation}
\label{equ:Bound_pmkm-1}
	\left|p^{(m-1)}_{m,k}(t)\right|
	\leq \frac{1}{C(k)} \sum^{k}_{n=1} \frac{1}{n^{m-\ell}} 
	\leq \frac{1}{C(k)} \sum^{k}_{n=1} \frac{1}{n^{}} 
	= 1\;,
\end{equation}
for all $t\in\TT$.
We also need a slightly better upper bound on $\big|p^{(m-1)}_{m,k}(t)\big|$.
To this end, we note first that \eqref{equ:deribative_pml_mMinus1} yields
\begin{equation*}
	\Big| p^{(m-1)}_{m,k}(t) \Big| = \tfrac{1}{C(k)} \Big| \textstyle\sum^{k}_{n=1} \frac{\cos(nt)}{n} \Big|\,.	
\end{equation*}
The trigonometric polynomial on the right hand side is not uniformly bounded with respect to $k$ on the whole interval $\TT$, but only on intervals of the form $[\epsilon, 2\pi-\epsilon] \subset \TT$.
The precise statement is given by Lemma~\ref{lem:AppUniformBound} in Appendix~\ref{app:Dirichlte}.
Using this lemma and the bound \eqref{equ:Bound_Ck} on $C(k)$, we get
\begin{equation}
\label{equ:UpperBound_pmmMinus1}
	\big| p^{(m-1)}_{m,k}(t) \big|
	\leq \frac{C_{2}(\epsilon)}{\log(k+1)}\,,
	\qquad\text{for all}\ t\in [\epsilon,2\pi - \epsilon]
\end{equation}
where $0 < \epsilon < \pi$ is an arbitrary real number.

Finally, we notice from \eqref{equ:deribative_pml_mMinus1} that $p^{(m-1)}_{m,k} \in \W(\TT)$ with
\begin{equation}
\label{equ:Wiener_pmk_mMinus1}
	\big\| p^{(m-1)}_{m,k} \big\|_{\W(\TT)} = 1\,,
	\qquad \text{for all}\ m,k\in\NN\,.
\end{equation}

\subsection{Properties of the functions $w_k$}
\label{app:Properties_wl}

Let $\alpha\in\Cweak$ be arbitrary and let $\left\{ \alpha_{n}\right\}^{\infty}_{n=0} \subset\RNc$ be a corresponding computable sequence as constructed in Subsection~\ref{app:LemmaSequence}.
For an arbitrary fixed $m\in\NN$, let $p_{m,k}$, $k\in\NN$ be the trigonometric polynomial defined in Subsection~\ref{app:TrigPolynomials}.
Therewith, we define for $\ell = 0,1,\dots,m$ the $2\pi$-periodic functions
\begin{equation}
\label{equ:Def_well_2}
	w_{\ell}(t) = \textstyle\sum^{\infty}_{k=1} \left(\alpha_{k} - \alpha_{k-1} \right) p^{(\ell)}_{m,k}(t)\,,
	\quad t\in [0,2\pi]\,,
\end{equation}
For us, the following properties of these functions will be of importance:

\begin{lemma}
\label{lem:Properties_w_ell}
For an arbitrary $\alpha\in\Cweak$ and a fixed $m\in\NN$, let $w_{\ell}$, $0 \leq \ell \leq m$ be the functions defined by \eqref{equ:Def_well_2} on $\TT = [0,2\pi]$.
Then
\begin{enumerate}
\item $w_{\ell} \in \Cc(\TT)$ for all $0 \leq \ell \leq m-2$.
\item If $w := w_{0}$, then $w^{(\ell)} = w_{\ell}$ for all $1\leq \ell \leq m-1$.
\item $w_{m-1} \in \W(\TT)$ is absolutely continuous.
\item $w_{m-1} \in \H12c(\TT)$
\item $w_{m-1}(t) \in \RNc$ for all $t\in (0,2\pi) \cap \RNc$
\item $w_{m-1}(0) = \alpha$.
\end{enumerate}
\end{lemma}

\begin{IEEEproof}
1)
Let $0 \leq \ell \leq m-2$ be arbitrary.
For every $R\in\NN$ and every $t\in\TT$, we have
\begin{align*}
	\left| w_{\ell}(t) - \textstyle\sum^{R}_{k=1} \left(\alpha_{k} - \alpha_{k-1} \right) p^{(\ell)}_{m,k}(t)\right|
	&\leq \textstyle\sum^{\infty}_{k=R+1} \big|\alpha_{k} - \alpha_{k-1} \big| \big| p^{(\ell)}_{m,k}(t) \big|\\
	&\leq \textstyle\frac{\pi^{2}}{6} \sum^{\infty}_{k=R+1} \big|\alpha_{k} - \alpha_{k-1} \big| \frac{1}{\log(k+1)}
	\leq \frac{\pi^{2}}{6} \frac{Q_{\alpha}}{\log(R+2)}
\end{align*}
using Point 3) of Lemma~\ref{lem:SeqAlpha} and \eqref{equ:Bound_pmkell} to obtain the last line.
This last inequality thus shows that the computable sequence
\begin{equation*}
	\left\{  \textstyle \sum^{R}_{k=1} \left(\alpha_{k} - \alpha_{k-1} \right) p^{(\ell)}_{m,k} \right\}_{R\in\NN}
\end{equation*}
of computable trigonometric polynomials converges effectively and uniformly to $w_{\ell}$ on $\TT$.
This proves that $w_{\ell} \in \Cc(\RN)$.

2) 
For an arbitrary $0 \leq \ell \leq m-2$, we set $f = w_{\ell}$, $g = w_{\ell+1}$ and for $R\in \NN$, we set
$\psi_{R} = \textstyle \sum^{R}_{k=1} \left(\alpha_{k} - \alpha_{k-1} \right) p^{(\ell)}_{m,k}$.
Then Part 1) of this proof shows that the functions $f,g$ and the sequence $\left\{ \psi_{R} \right\}_{R\in\NN}$ satisfies the conditions of Lemma~\ref{lem:AppUniformConv}.
This lemma implies that $w'_{\ell} = w_{\ell+1}$ and consequently proves Statement 2). 
In particular, we have $w^{(\ell)} \in \Cc(\RN)$ for all $0 \leq \ell \leq m-2$.

3)
To simplify notation, we define for $R\in\NN$
\begin{equation}
\label{equ:varphiR}
	\varphi_{R}(t) = \textstyle\sum^{R}_{k=1} \left(\alpha_{k} - \alpha_{k-1}\right) p^{(m-1)}_{m,k}(t)\,,
	\qquad t\in\TT\,.
\end{equation}
Then for all $R_{2} > R_{1} \geq 1$, we have
\begin{equation*}
	\big\| \varphi_{R_2} - \varphi_{R_1}\big\|_{\W(\TT)}
	= \big\| \textstyle \sum^{R_2}_{k=R_1+1} \left(\alpha_{k} - \alpha_{k-1}\right) p^{(m-1)}_{m,k}\big\|_{\W(\TT)}
	\leq \textstyle\sum^{R_2}_{k=R_1+1} \left| \alpha_{k} - \alpha_{k-1}\right|
\end{equation*}
using the triangle inequality and \eqref{equ:Wiener_pmk_mMinus1}.
According to Lemma~\ref{lem:SeqAlpha}, Point 2), the right hand side converges to zero as $R_{1}$ and $R_{2}$ go to infinity.
Thus $\left\{ \varphi_{R} \right\}_{R\in\NN}$ is a Cauchy sequence in the Wiener algebra $\W(\TT)$. 
Since $\A(\TT)$ is complete, we have
\begin{equation*}
	w_{m-1} = \textstyle\lim_{R\to\infty} \sum^{R}_{k=1} \left(\alpha_{k} - \alpha_{k-1}\right) p^{(m-1)}_{m,k} \in \A(\TT)\,.
\end{equation*}
Moreover, similarly as in Part 1) of this proof, we have for arbitrary $R\in\NN$ and every $t\in\TT$ 
\begin{equation}
\label{equ:UniformConv_w}
	\big| w_{m-1}(t) - \varphi_{R}(t)\big|
	\leq \textstyle\sum^{\infty}_{k=R+1} \big|\alpha_{k} - \alpha_{k-1} \big| \big| p^{(m-1)}_{m,k}(t) \big|
	\leq \textstyle\sum^{\infty}_{k=R+1} \big|\alpha_{k} - \alpha_{k-1} \big|
\end{equation}
using \eqref{equ:Bound_pmkm-1} to obtain the last inequality.
The right hand side converges to zero as $R\to\infty$ by Point 2) of Lemma~\ref{lem:SeqAlpha}, showing that the computable sequence $\left\{  \varphi_{R} \right\}_{R\in\NN}$
of computable trigonometric polynomials converges uniformly to $w_{m-1}$ on $\TT$.
Consequently Lemma~\ref{lem:AppUniformConv} implies $w_{m-1} \in \C(\TT)$ and that $w_{m-1} = w'_{m-2}$. So in particular $w_{m-1} = w^{(m-1)}$.

To verify that $w^{(m-1)}$ is absolute continuous, we need to show that $w^{(m)}$ exists almost everywhere on $\TT$ and that
\begin{equation}
\label{equ:AbsolutCont_w_mMinus1}
	w^{(m-1)}(t) = w^{(m-1)}(0) + \int^{t}_{0} w^{(m)}(\tau)\,\d\tau\,,
	\quad t\in\TT\;.
\end{equation}
To this end, we note from \eqref{equ:deribative_pml_mMinus1} that 
\begin{equation*}
	p^{(m)}_{m,k}(t)
	= \textstyle-\frac{1}{C(k)} \sum^{k}_{n=1} \sin(nt)
	= -\frac{1}{C(k)} \widetilde{D}_{k}(t)\,,
	\quad t\in\TT\;.
\end{equation*}
where $\widetilde{D}_{k}$ denotes the conjugate Dirichlet's  kernel (cf. Appendix~\ref{app:Dirichlte}).
By the properties of this $\widetilde{D}_{k}$ as given in Lemma~\ref{lem:L1Dirichlet}, we have
\begin{equation*}
	\big\| p^{(m)}_{m,k} \big\|_{1}
	= \frac{1}{C(k)} \big\| \widetilde{D}_{k} \big\|_{1}
	\leq \frac{1}{\log(k+1)} \left[C_{3} + C_{4}\log(k)\right]
	\leq \frac{C_{3}}{\log 2} + C_{4} =: C_{5}\,,
	\qquad \text{for all}\ k\in\NN\,,
\end{equation*}
with positive constants $C_{3}, C_{4}$ that do not depend on $k$.

Let $R_{1}, R_{2}\in\NN$ with $R_{2} > R_{1}$ be arbitrary.
Then
\begin{equation*}
	\left\|\varphi'_{R_2} - \varphi'_{R_1}\right\|_{1}
	= \left\| \textstyle\sum^{R_2}_{k=R_{1}+1} \left(\alpha_{k} - \alpha_{k-1}\right) p^{(m)}_{m,k} \right\|_{1}
	\leq \textstyle\sum^{R_2}_{k=R_{1}+1} \left|\alpha_{k} - \alpha_{k-1}\right| \big\|p^{(m)}_{m,k}\big\|_{1}
	\leq C_{5} \textstyle\sum^{R_2}_{k=R_{1}+1} \left|\alpha_{k} - \alpha_{k-1}\right|
\end{equation*}
The right hand side of the last inequality converges to zero as $R_{1}, R_{2} \to\infty$, by Point 2) of Lemma~\ref{lem:SeqAlpha}. 
This shows that $\left\{ \varphi'_{R} \right\}_{R\in\NN}$ is a Cauchy sequence in $L^{1}(\TT)$ and so the limit
\begin{equation*}
	w_{m} = w^{(m)} = \textstyle\sum^{\infty}_{k=1} \left(\alpha_{k} - \alpha_{k-1}\right) p^{(m)}_{m,k}
\end{equation*}
is an $L^{1}(\TT)$-function.
Finally, we proceed similarly as in the proof of Lemma~\ref{lem:AppUniformConv}:
\begin{multline}
\label{equ:wm}
	\left| \int^{t}_{0} w^{(m)}(\tau)\, \d\tau - \left[w^{(m-1)}(t) - w^{(m-1)}(0) \right]\right|\\
	= \left| \int^{t}_{0} \left[ w^{(m)}(\tau) - \varphi'_{R}(\tau) \right] \d\tau + \left[ \varphi_{R}(t) - w^{(m-1)}(t)\right] - \left[\varphi_{R}(0) - w^{(m-1)}(0) \right] \right|\\
	\leq \big\| w^{(m)} - \varphi'_{R} \big\|_{1} + 2\, \big\| w^{(m-1)} - \varphi_{R} \big\|_{\infty}\,.
\end{multline}
We already showed that $\left\{ \varphi_{R} \right\}_{R\in\NN}$ converges to $w^{(m-1)}$ uniformly on $\TT$ and that $\left\{ \varphi'_{R} \right\}_{R\in\NN}$ converges to $w^{(m)}$ in $L^{1}(\TT)$.
Therefore, the right hand side of \eqref{equ:wm} converges to zero as $R\to\infty$, proving \eqref{equ:AbsolutCont_w_mMinus1}.

4) We consider the trigonometric polynomial $p^{(m-1)}_{m,k}$ given by \eqref{equ:deribative_pml_mMinus1}. 
Its Fourier coefficients are $a_{n} = \frac{1}{n\, C(k)}$, $n = 1,2,\dots,k$ and so its Dirichlet energy satisfies
\begin{equation*}
	\left\| p^{(m-1)}_{m,k} \right\|_{H^{1/2}}
	= \left( \textstyle\frac{1}{2}\sum^{k}_{n=1} n\, \left|a_{n}\right|^{2} \right)^{1/2}
	= \frac{1}{C(k)} \left( \textstyle\frac{1}{2}\sum^{k}_{n=1} \frac{1}{n} \right)^{1/2}
	= \frac{1}{\sqrt{2\, C(k)}}
	\leq \frac{1}{\sqrt{2\,\log(k+1)}}\,.
\end{equation*}
Now we show that the sequence $\left\{ \varphi_{R} \right\}_{R\in\NN}$ of the computable trigonometric polynomials \eqref{equ:varphiR} converge effectively to $w_{m-1}$ in the $H^{1/2}$-norm as $R\to\infty$.
Indeed,
\begin{multline*}
	\Big\| w^{(m-1)} - \varphi_{R} \Big\|_{H^{1/2}}
	= \Big\|\sum^{\infty}_{k=R+1} \left(\alpha_{k} - \alpha_{k-1}\right) p^{(m-1)}_{m,k}\Big\|_{H^{1/2}}
	\leq \sum^{\infty}_{k=R+1} \left| \alpha_{k} - \alpha_{k-1} \right| \big\| p^{(m-1)}_{m,k}\big\|_{H^{1/2}}\\
	\leq \frac{1}{\sqrt{2\, \log(R+2)}} \sum^{\infty}_{k=R+1} \left| \alpha_{k} - \alpha_{k-1} \right|
	\leq \frac{Q_{\alpha}}{\sqrt{2\, \log(R+2)}}\;,
\end{multline*}
which shows that $\varphi_{R}$ effectively converges to $w^{(m-1)}$ in the $H^{1/2}$-norm as $R\to\infty$ and proves that $w^{(m-1)}$ is a computable function in $H^{1/2}(\TT)$.

5)
Note that the uniform convergence proved in \eqref{equ:UniformConv_w} is not effective. 
However, using Lemma~\ref{lem:AppUniformBound}, we can show that for every $\epsilon> 0$
the sequence $\left\{ \varphi_{R} \right\}_{R\in\NN}$ of trigonometric polynomials defined in \eqref{equ:varphiR} converges effectively (and uniformly) to $w_{m-1}$ on every interval $[\epsilon, 2\pi-\epsilon]$.
Indeed, using triangle inequality and \eqref{equ:UpperBound_pmmMinus1}, we get
\begin{multline*}
	\left|w_{m-1}(t) - \varphi_{R}(t)\right| = 
	 \left| \textstyle\sum^{\infty}_{k=R+1} \left( \alpha_{k} - \alpha_{k-1} \right)  p^{(m-1)}_{m,k}(t) \right|
	\leq \textstyle\sum^{\infty}_{k=R+1} \big| \alpha_{k} - \alpha_{k-1} \big| \big| p^{(m-1)}_{m,k}(t) \big|\\
	\leq \textstyle\sum^{\infty}_{k=R+1} \big| \alpha_{k} - \alpha_{k-1} \big| \frac{1}{C(k)} C_{\epsilon}
	\leq \tfrac{C_{2}(\epsilon)\cdot Q_{\alpha}}{C(R+1)}
	\leq \tfrac{C_{2}(\epsilon)\cdot Q_{\alpha}}{\log(R+2)}\,,
\end{multline*}
where in the last line, we also used Property 3) of Lemma~\ref{lem:SeqAlpha}.
Thus we showed that for every $0 < \epsilon <\pi$, the function $w_{m-1} = w^{(m-1)}$ is a computable continuous functions on $[\epsilon,2\pi-\epsilon]$.
This implies in particular that $w^{(m-1)}(t) \in \RNc$ for all $t \in [\epsilon,2\pi-\epsilon] \cap \RNc$. Since this holds for arbitrary $\epsilon > 0$, we thus have
\begin{equation*}
	w^{(m-1)}(t) \in \RNc\,,
	\qquad\text{for all}\ t \in (0,2\pi) \cap \RNc\;.
\end{equation*}

6) Now we consider $t=0$.
By the definition of $w_{m-1}$ in \eqref{equ:Def_well_2}, we have
\begin{equation*}
	w_{m-1}(0) = \textstyle\sum^{\infty}_{k=1} \left(\alpha_{k} - \alpha_{k-1} \right) p^{(m-1)}_{m,k}(0)\;.
\end{equation*}
Moreover, \eqref{equ:deribative_pml_mMinus1} shows that $p^{(m-1)}_{m,k}(0) = 1$ and so we get
\begin{equation*}
	w_{m-1}(0)
	= \textstyle\lim_{N\to\infty} \sum^{N}_{k=1} \left(\alpha_{k} - \alpha_{k-1}\right)
	= \lim_{N\to\infty} \alpha_{N} - \alpha_{0} = \alpha\;,
\end{equation*}
proving Point 6) of Lemma~\ref{lem:Properties_w_ell}.
\end{IEEEproof}

\subsection{Proof of Lemma~\ref{lem:MainLemma}}
\label{sec:ProofMainLemma}

We fix a positive $m\in\NN$, an arbitrary $\delta\in\RNc$, and an arbitrary $\alpha_{*} \in \Cweak$.
Then we define
\begin{equation}
\label{equ:DefAlpha}
	\alpha = \alpha_{*}\left( \tfrac{\delta}{2\pi} \right)^{m-1}\,.
\end{equation}
Since $(\delta/(2\pi))^{m-1}$ is a computable number, it follows that $\alpha\in\Cweak$.
With this $\alpha$ we associate a computable sequence of computable numbers $\left\{ \alpha_{n} \right\}^{\infty}_{n=0}$ whose properties are given in Lemma~\ref{lem:SeqAlpha} in Subsection~\ref{app:LemmaSequence}. 
Based on this sequence and using the trigonometric polynomials defined in Subsection~\ref{app:TrigPolynomials}, we define for $0 \leq \ell \leq m-1$ the $2\pi$-periodic functions given in \eqref{equ:Def_well_2}, i.e.
\begin{equation}
\label{equ:Def_well}
	w_{\ell}(t) = \textstyle\sum^{\infty}_{k=1} \left(\alpha_{k} - \alpha_{k-1} \right) p^{(\ell)}_{m,k}(t)\,,
	\qquad t\in \RN\;,
\end{equation}
wherein $p^{(\ell)}_{m,k}$ denotes the $\ell$-th derivative of the trigonometric polynomial $p_{m,k}$ defined by \eqref{equ:definition_pmk} in Subsection~\ref{app:TrigPolynomials} dependent on the actual $m$.
Properties of the so defined functions are given in Lemma~\ref{lem:Properties_w_ell}.
This lemma shows in particular, that if we set $w = w_{0}$, then
\begin{equation*}
	w^{(\ell)} = w_{\ell}\,,\quad \text{for all}\ 0\leq \ell \leq m-1\,,
\end{equation*}
that all $w^{(l)}$ with $0\leq \ell \leq m-2$ are computable continuous functions on $[0,2\pi]$,
and that $w^{(m-1)}\in \A([0,2\pi])$ is absolute continuous with $w^{(m-1)} \in \H12c([0,2\pi])$.

Based on this function $w = w_{0} \in \C^{{m-1}}(\RN)$, we define now the function
\begin{equation*}
	u(t) = w\left(\tfrac{2\pi}{\delta} t - \pi\right)\,,\qquad t\in\TT = [0,\delta]\,.
\end{equation*}
Then it is clear that for all $0\leq \ell\leq m-1$
\begin{equation*}
	u^{(\ell)}(t) = \left(\tfrac{2\pi}{\delta}\right)^{\ell}\, w^{(\ell)}\left(\tfrac{2\pi}{\delta} t - \pi\right)\,,
	\qquad t\in[0,2\pi]\,,
\end{equation*}
and Properties 1) -- 4) follow directly from the properties of $w$ given in Lemma~\ref{lem:Properties_w_ell}.
Moreover, we also have
\begin{equation*}
	u^{(m-1)}(0) = \left(\tfrac{2\pi}{\delta}\right)^{m-1} w(-\pi) = \left(\tfrac{2\pi}{\delta}\right)^{m-1} w(\pi) \in \RNc
\end{equation*}
by Property 5) of Lemma~\ref{lem:Properties_w_ell} and because $\left(\tfrac{2\pi}{\delta}\right)^{m-1} \in \RNc$.
Finally
\begin{equation*}
	u^{(m-1)}\left(\tfrac{\delta}{2}\right)
	= \left(\tfrac{2\pi}{\delta}\right)^{m-1} w^{(m-1)}(0)
	= \left(\tfrac{2\pi}{\delta}\right)^{m-1} \alpha
	= \alpha_{*}
\end{equation*}
using Property 6) of Lemma~\ref{lem:Properties_w_ell} and recalling that $\alpha$ was defined by \eqref{equ:DefAlpha}.


\section{Conclusions}
\label{sec:Conclusion}

In this paper, we have investigated to computability of solutions of singular linear differential-algebraic equations (LDAEs) with constant coefficients.
Assuming that all coefficients of the LDAE are computable numbers, we have shown that for arbitrary computable initial values there always exist admissible and computable input functions such that the LDAE has a unique solution, yet this solution is not computable. In fact the solution is even not Banach-Mazur computable on arbitrary small time intervals.
This non-computability of the solution occurs even through the input $\vu$ has strong regularity properties.
Namely, if the given LDAE has index $m$, then the non-computability occurs for a computable input $\vu$ that is $m-1$-times continuously differentiable and for which its $(m-1)$th derivative $\vu^{(m-1)}$ is absolute continuous, admits an absolutely convergent Fourier series and has a computable finite Dirichlet energy.

This behavior is quite remarkable in view of known computability results for regular LDAEs.
There it is known that if the LDAE has a unique solution, then this solution is necessarily Turing computable.

%

\appendix

\subsection{The matrices $\widetilde{\mE}$, $\widetilde{\mA}$, $\mP$, and $\mQ$ are computable}
\label{app:ComputablePQ}

Let $(\mE,\mA)$ be a regular matrix pair with computable matrices $\mE,\mA \in \RNc^{n \times n}$,
and let $\mP$ and $\mQ$ be the matrices in $\RN^{n\times n}$ that transform $(\mE,\mA)$ onto its Weiserstrass canonical form \eqref{equ:MazrixWeierstrassCF} via \eqref{equ:TrafoMatrices}.
We want to show that $\mP$ and $\mQ$ are also computable matrices, i.e. all of its entries are again computable numbers, i.e. $\mP, \mQ \in \RNc^{n\times n}$.
Thereby, we follow the construction of the canonical form in \cite{Gantmacher1} but only sketch the main arguments. 

To this end, we recall first that two quadratic matrices $\mB_{1}, \mB_{2} \in \FN^{n\times n}$, with entries from an arbitrary number field $\FN$, are said to be \emph{similar}, if there exists an invertible matrix $\mP \in \FN^{n\times n}$ such that $\mB_{2} = \mP^{-1} \mB_{1} \mP$.
The important point for us is that the transformation matrix of two similar matrices has always entries from the same field as the matrices $\mB_{1}$ and $\mB_{2}$ itself.
In our case, the field $\FN$ will be the field of computable numbers $\RNc$.

However, we are interested in the similarity of matrix pairs and not just matrices.
Therefore, we follow \cite{Gantmacher1} and consider the matrix polynomial (or matrix pencil)
\begin{equation}
\label{equ:MatrixPencil}
	\lambda \mE - \mA
\end{equation}
associated with the matrix pair $(\mE,\mA)$.

1) Assume first that $\mE$ is non-singular. Then we can multiply \eqref{equ:MatrixPencil} from the left with $\mE^{-1}$ and obtain $\lambda\mI - \mA_{1}$ with $\mA_{1} = \mE^{-1}\mA$ and it is clear that all entries of $\mE^{1}$ and consequently of $\mA_{1}$ are computable numbers, i.e. $\mA_{1} \in \RNc^{n\times n}$.
Now, every matrix is similar to its Jordan canonical form. Consequently, there exists $\mP_{1} \in \RNc^{n\times n}$ such that $\mP_{1}^{-1}\mA_{1}\mP_{1} = \mJ$, where $\mJ$ is the Jordan canonical form of $\mA_{1}$.
So overall, we get
\begin{equation*}
	\mP^{-1}_{1}\mE^{-1}\left( \lambda \mE - \mA\right) \mP_{1}
	= \mP^{-1} \left( \lambda \mE - \mA\right) \mQ
	= \lambda \mI - \mJ\;.
\end{equation*}
The right hand side is the desired canonical form \eqref{equ:MazrixWeierstrassCF} with $\nb = 0$ because $\mE$ was assumed to be regular.
The corresponding transformation matrices $\mP = \mE\mP_{1}$ and $\mQ = \mP_{1}$ are computable, because $\mE$ and $\mP_{1}$ are computable matrices.

2) Assume now that $\mE$ is singular. Since $(\mE,\mA)$ is regular, there always exists an $c \in \RNc$ such that $\det(c\mE - \mA) \neq 0$.
With this $c$, we rewrite the matrix polynomial \eqref{equ:MatrixPencil} as
\begin{equation*}
	\lambda\mE - \mA = (\lambda - c)\mE - \mA_{1}
\end{equation*}
with $\mA_{1} = \mA - c\mE \in \RNc^{n\times n}$ and $\det(\mA_{1})\neq 0$.
Now we multiply the matrix pencil with $\mA^{-1}_{1} \in \RNc^{n\times n}$ which yields
\begin{equation*}
	\mA^{-1}_{1}\left( \lambda\mE - \mA \right) = (\lambda - c)\mE_{1} - \mI
\end{equation*}
with $\mE_{1} = \mA^{-1}_{1}\mE \in \RNc^{n\times n}$.
Matrix $\mE_{1}$ is singular, but as any matrix, it is similar to its Jordan canonical form $\mJ$, i.e. there exists $\mP_{1}\in\RN^{n\times n}$ such that $\mP^{-1}_{1}\mE_{1}\mP_{1} = \mJ$.
Therewith, the previous matrix pencil can be transformed as
\begin{equation*}
	\mP^{-1}_{1}\mA^{-1}_{1}\left( \lambda\mE - \mA \right)\mP_{1} = (\lambda - c)\mJ - \mI\;.
\end{equation*}
Since $\mE_{1}$ is singular, its Jordan normal form has the block diagonal form
\begin{equation*}
	\mJ
	= \begin{pmatrix}
	\mJ_{\mathrm{a}} & \mNull\\
	\mNull & \mJ_{\mathrm{b}}
	\end{pmatrix}
\end{equation*}
wherein $\mJ_{\mathrm{a}} \in \RNc^{\na\times\na}$ is invertible, $\mJ_{\mathrm{b}} \in \RNc^{\nb\times\nb}$ is a nilpotent, and $\na+\nb = n$.
Therewith, we get
\begin{equation}
\label{equ:MatrixPencil_Block}
	\mP^{-1}_{2}\left( \lambda\mE - \mA \right)\mQ_{2}
	= \begin{pmatrix}
	(\lambda - c)\mJ_{\mathrm{a}} - \mI & \mNull\\
	\mNull & \lambda\mJ_{\mathrm{b}} - c(\mI + \mJ_{\mathrm{b}})
	\end{pmatrix}\,,
\end{equation}
where we defined $\mP_{2} = \mA_{1}\mP_{1} \in \RNc^{n\times n}$ and $\mQ_{2} = \mP_{1} \in \RNc^{n\times n}$.
Thus the matrix pencil separates into two smaller independent matrix pencils which are we consider separately next.

Since $\mJ_{\mathrm{a}}$ is non-singular, we can multiply the first matrix polynomial from the left with $\mJ^{-1}_{\mathrm{a}} \in \RNc^{\na\times\na}$ which yields
$\lambda\mI - \mK_{\mathrm{a}}$ with $\mK_{\mathrm{a}} = \left(c\mI + \mJ^{-1}_{\mathrm{a}}\right)\in\RNc^{\na\times\na}$.
Finally, we can transform the invertible $\mK_{0}$ into its Jordan canonical form $\mJ$ according to $\mP^{-1}_{3}\mK_{0}\mP_{3} = \mJ$ with some invertible transformation matrix $\mP_{3} \in \RNc^{\na\times\na}$. Thus overall, we have for the first line
\begin{equation*}
	\mP_{4}^{-1} \left[(\lambda - c)\mJ_{\mathrm{a}} - \mI\right]\mQ_{4}
	= \lambda\mI - \mJ
\end{equation*}
with $\mP_{4} = \mJ_{\mathrm{a}}\mP_{3} \in \RNc^{\na\times\na}$ and $\mQ_{4} = \mP_{3}$.

For the second line in \eqref{equ:MatrixPencil_Block}, we notice first that since $\mJ_{\mathrm{b}}$ is nilpotent, the matrix $\mI+\mJ_{\mathrm{b}}$ is always invertible.
Thus, multiplying the second line of \eqref{equ:MatrixPencil_Block} from the left with $[c(\mI+\mJ_{\mathrm{b}})]^{-1}$ yields
$\lambda\mK_{\mathrm{b}} - \mI$ with $\mK_{\mathrm{b}} = [c(\mI+\mJ_{\mathrm{b}})]^{-1} \mJ_{\mathrm{b}}$.
Finally, $\mK_{\mathrm{b}}$ is similar to its Jordan canonical form $\mN$ so there exist a corresponding transformation matrix $\mP_{5} \in \RNc^{\nb\times\nb}$ and therewith we get for the second line of \eqref{equ:MatrixPencil_Block}
\begin{equation*}
	\mP_{6}^{-1} \left[ \lambda\mJ_{\mathrm{b}} - c(\mI + \mJ_{\mathrm{b}}) \right] \mQ_{6}
	= \lambda\mN - \mI
\end{equation*}
with $\mP_{6} = c(\mI + \mJ_{\mathrm{b}})\mP_{5} \in \RNc^{\na\times\na}$ and $\mQ_{6} = \mP_{5}$.

Putting all together, we thus have found the transformation matrices $\mP, \mQ \in \RNc^{n\times n}$, given by
\begin{equation*}
	\mP = \mP_{2}\begin{pmatrix}\mP_{4} & \mNull\\ \mNull & \mP_{6}\end{pmatrix}
	\quad\text{and}\quad
	\mQ = \mQ_{2}\begin{pmatrix}\mQ_{4} & \mNull\\ \mNull & \mQ_{6}\end{pmatrix}
\end{equation*}
and such that
\begin{equation*}
	\mP^{-1}\left(\lambda\mE - \mA\right)\mQ 
	= \lambda\begin{pmatrix} \mI & \mNull\\ \mNull & \mN\end{pmatrix} - \begin{pmatrix} \mJ & \mNull\\ \mNull & \mI\end{pmatrix}
	= \lambda\widetilde{\mE} - \widetilde{\mA}
\end{equation*}
where the right hand side of the last equation is the desired Weierstrass canonical form \eqref{equ:MazrixWeierstrassCF}, and the previous derivation showed that all entries of $\mP$ and $\mP$ are in $\RNc$.

To verify that also the entries of $\widetilde{\mE}$ and $\widetilde{\mA}$ are computable numbers, we notice first that $\mN$ is the Jordan canonical form of a nilpotent matrix. 
Therefore, it has the form \eqref{equ:Nk}. 
Therefore, all entries of $\widetilde{\mE}$ are either zero or one, which shows that $\widetilde{\mE} \in \RNc^{n\times n}$.
The matrix $\mJ$ within $\widetilde{\mA}$ is the Jordan normal form of a regular matrix which contains (apart form zero and ones) the finite eigenvalues of $(\mE,\mA)$ on its diagonal (see, e.g., \cite{Gantmacher1} for details).
These finite eigenvalues are the zeros of the characteristic polynomial $p(\lambda) = \det(\lambda\mE - \mA)$.
Since $\mE$ and $\mA$ are computable numbers, it follows that the coefficients of the characteristic polynomial are also computable numbers.
Consequently and because the computable numbers are algebraic closed, all of its zeros are also computable (complex) numbers \cite{Ko_Complexity91} which shows that also $\widetilde{\mA} \in \RNc$.

\subsection{Dirichlets's  kernels}
\label{app:Dirichlte}

Consider the interval $\TT = [0,2\pi] \subset\RN$ of the real axis. The trigonometric polynomials
\begin{equation*}
	D_{n}(t)
	= \frac{1}{2} + \sum^{n}_{k=1}\cos(kt)
	= \frac{\sin([n+1/2])t}{2\sin(t/2)}\,,
	\quad t\in\TT\,,
\end{equation*}
and
\begin{equation*}
	\widetilde{D}_{n}(t)
	= \sum^{n}_{k=1} \sin(k t)
	= \frac{\cos(t/2) - \cos([n+1/2]t)}{2\sin(t/2)}\,,
	\quad t\in\TT\,,	
\end{equation*}
are known \emph{Dirichlet's kernel} and \emph{Dirichlet's conjugate kernel}, respectively.
We need two well known statements related to these kernels, namely Lemmas~\ref{lem:AppUniformBound} and \ref{lem:L1Dirichlet}, below.
For completeness and to be self contained, we also provide short proofs of the results which are mainly taken from \cite{Zygmund}.

First we need the observation that both kernels are uniformly bounded with respect to $n\in\NN$, on every interval of the form $[\epsilon,2\pi-\epsilon] \subset\TT$:

\begin{lemma}
\label{lem:DUniformUpBound}
Let $\epsilon > 0$ be arbitrary, then there exists a constant $C_{1}(\epsilon) > 0$ such that for every $n\in\NN$ and all $t\in [\epsilon,2\pi-\epsilon]$, we have
\begin{equation*}
	\big|D_{n}(t)\big| \leq C_{1}(\epsilon)
	\quad\text{and}\quad
	\big|\widetilde{D}_{n}(t)\big| \leq C_{1}(\epsilon)\,.
\end{equation*}
\end{lemma}
To prove this statement, one only has to note that the nominators of the closed form expressions for $D_{n}$ and $\widetilde{D}_{n}$ are uniformly bounded whereas the denominator satisfies $\sin(t/2) \geq \sin(\epsilon/2)$ for all $t\in [\epsilon,2\pi-\epsilon]$.

Based on the previous lemma, one can now show \cite{Zygmund}:

\begin{lemma}
\label{lem:AppUniformBound}
To every $\epsilon > 0$ there exists a constant $C_{2}(\epsilon)$ such that for all $k\in\NN$, we have
\begin{equation*}
	\textstyle\left|\sum^{k}_{n=1} \frac{1}{n} \cos(nt)\right| \leq C_{2}(\epsilon) \qquad\text{and}\qquad 
	\textstyle\left|\sum^{k}_{n=1} \frac{1}{n} \sin(nt)\right| \leq C_{2}(\epsilon)\,, \qquad\text{for all}\ t\in [\epsilon, 2\pi-\epsilon]\,.
\end{equation*} 
\end{lemma}

\begin{IEEEproof}
To proof this statement, we use summation by parts (cf. \cite[Capter~I.2]{Zygmund}), which shows that
\begin{equation*}
	\left| \textstyle\sum^{k}_{n=1} \frac{1}{n}\cos(nt) \right|
	= \left| \textstyle\sum^{k-1}_{n=1} \left[D_{n}(t)-\tfrac{1}{2}\right]\,\big(\tfrac{1}{n} - \tfrac{1}{n+1}\big) + \left[D_{k}(t) - \tfrac{1}{2}\right]\tfrac{1}{k} \right|
\end{equation*}
Choose an arbitrary $t\in [\epsilon,2\pi-\epsilon]$, applying triangle inequality, and using Lemma~\ref{lem:DUniformUpBound}, one obtains
\begin{equation*}
	\left|\textstyle\sum^{k}_{n=1} \tfrac{1}{n}\cos(nt)\right|
	\leq \left[C_{1}(\epsilon)+\tfrac{1}{2}\right] \left(\textstyle\sum^{k-1}_{n=1} \left[\tfrac{1}{n}-\tfrac{1}{n+1}\right] + \tfrac{1}{k}\right)
	= C_{1}(\epsilon)+\tfrac{1}{2}.
\end{equation*}
which proves the statement for the Dirichlet kernel.
The proof for the conjugate Dirichlet kernel is exactly the same.
\end{IEEEproof}

In this paper, will will need an upper bound on its $L^{1}(\TT)$-norm. It is well known that this norm grows logarithmically in $n$ \cite[Chapter~II.12]{Zygmund}.
More precisely, one has the following result:

\begin{lemma}
\label{lem:L1Dirichlet}
There exists positive constants $C_{3}$ and $C_{4}$ such that for all $n\in\NN$
\begin{equation}
\label{equ:BoundDirichlet}
	\big\|\widetilde{D}_{n}\big\|_{1}
	= \int^{2\pi}_{0} \big|\widetilde{D}_{n}(t)\big|\, \d t 
	\leq C_{3} + C_{4} \log(n)\,.
\end{equation}
\end{lemma}

For completeness and to be self-contained, we include the standard proof.

\begin{IEEEproof}
It will be convenient to consider first the \emph{modified conjugate Dirichlet kernel} given by
\begin{equation}
\label{equ:ModConjDirichKern}
	\widetilde{D}^{*}_{n}(t)
	= \widetilde{D}_{n}(t) - \tfrac{1}{2}\sin(nt)
	= \frac{[1 - \cos(nt)]\cos(t/2)}{2 \sin(t/2)}\,.
\end{equation}
Because $\widetilde{D}^{*}_{n}$ is a $2\pi$-periodic and odd function (i.e. $\widetilde{D}^{*}_{n}(-t) = - \widetilde{D}^{*}_{n}(t)$) for all $t\in\RN$, we have
\begin{equation*}
	\big\|\widetilde{D}^{*}_{n}\big\|_{1}
	= \int^{2\pi}_{0} \big|\widetilde{D}^{*}_{n}(t)\big|\, \d t 
	= 2\int^{\pi}_{0} \big|\widetilde{D}^{*}_{n}(t)\big|\, \d t \;.
\end{equation*}
Inserting the closed form expression \eqref{equ:ModConjDirichKern} for $\widetilde{D}^{*}_{n}$, using that $\left|\cos(t/2)\right|\leq 1$, 
and $\sin(t/2) \geq \frac{1}{\pi} t$ for all $t\in [0,\pi]$,
we get
\begin{equation*}
	\big\|\widetilde{D}^{*}_{n}\big\|_{1}	
	\leq \pi \int^{\pi}_{0} \frac{\left|1 - \cos(nt)\right|}{t}\, \d t
	= \pi \int^{n \pi}_{0} \frac{1 - \cos(t)}{t}\, \d t
	= \pi \int^{\pi}_{0} \frac{1 - \cos(t)}{t}\, \d t + \pi \sum^{n-1}_{k=1}\int^{(k+1)\pi}_{k\pi} \frac{1 - \cos(t)}{t}\, \d t\,.
\end{equation*}
To upper bound the first term on the right hand side, we use that $1 - \cos(t) \leq \frac{1}{2} t^{2}$.
For the second term, we use the trivial observation that $\frac{1}{t} \leq \frac{1}{k\pi}$ for all $t \in [k\pi,(k+1)\pi]$.
This yields
\begin{equation*}
	\big\|\widetilde{D}^{*}_{n}\big\|_{1}
	\leq \frac{\pi}{2} \int^{\pi}_{0} t\, \d t + \pi \sum^{n-1}_{k=1}\frac{1}{k}
	= \frac{\pi^{3}}{4} + \pi + \pi \sum^{n-1}_{k=2}\frac{1}{k}
	\leq \pi\left( \tfrac{\pi^{2}}{4} + 1 + \int^{n-1}_1\frac{\d\tau}{\tau}\right)
	\leq \widetilde{C}_{3} + \pi\log(n-1)
\end{equation*}
with $\widetilde{C}_{3} = \pi\left( \tfrac{\pi^{2}}{4} + 1\right)$.
Finally, applying triangle inequality to \eqref{equ:ModConjDirichKern} gives
\begin{equation*}
	\big\|\widetilde{D}_{n}\big\|_{1} 
	\leq \big\|\widetilde{D}^{*}_{n}\big\|_{1} + \tfrac{1}{2} \big\|\sin(n\cdot)\big\|_{1}
\end{equation*}
which yields \eqref{equ:BoundDirichlet} with $C_{3} = \pi\left( \tfrac{\pi^{2}}{4} + 1\right) + 2$ and $C_{4} = \pi$.
\end{IEEEproof}

\begin{remark}
The constants $C_{1}$ and $C_{2}$ derived in the previous proof are not optimal.
\end{remark}

\subsection{Some Auxiliary Results}
\label{app:AuxResults}

\begin{lemma}
\label{lem:AppUniformConv}
Let $f,g$ be functions on $\TT = [0,T] \subset \RN$, and let $\left\{ \psi_{n} \right\}_{n\in\NN} \subset \C^{1}(\TT)$ be a sequence of continuously differentiable functions on $\TT$ such that
\begin{align}
\label{equ:appUnifomConv}
\begin{split}
	&1)\qquad \lim_{n\to\infty} \max_{t\in\TT} \left| f(t) - \psi_{n}(t)\right| = 0\\
	&2)\qquad \lim_{n\to\infty} \max_{t\in\TT} \left| g(t) - \psi'_{n}(t)\right| = 0
	\end{split}
\end{align}
Then $f,g\in\C(\TT)$ and $f$ is continuously differentiable with $f'(t) = g(t)$, $t\in\TT$.
\end{lemma}

\begin{IEEEproof}
The continuity of $f$ and $g$ follows from the uniform convergence of sequences of continuous functions on $\TT$ \cite{Rudin1}.
We show next that 
\begin{equation}
\label{eqe:App_fStich_g}
	f(t) = f(0) + \int^{t}_{0}g(\tau)\, \d\tau\;,
	\qquad t\in \TT\,.
\end{equation}
This implies that $g = f'$, and since $g$ is continuous, it follows that $f \in \C^{1}(\TT)$.

Let $t \in\TT$ be arbitrary, then by triangle inequality, we have
\begin{multline*}
	\left| \int^{t}_{0} g(\tau)\, \d\tau - \left[ f(t) - f(0) \right] \right|
	\leq \left| \int^{t}_{0} g(\tau)\, \d\tau - \int^{t}_{0} f'_{n}(\tau)\d\tau \right| + \left| \int^{t}_{0} f'_{n}(\tau)\d\tau - \left[ f(t) - f(0) \right] \right|\\
	\leq \int^{t}_{0} \left|g(\tau) - f'_{n}(\tau)\left| \d\tau + \left| f'_{n}(t) - f(t) \right| + \right| f_{n}(0) - f(0)\right|
	\leq t\, \left\|g - f_{n} \right\|_{\infty} + 2\, \left\|f - f_{n} \right\|_{\infty}\,.
\end{multline*}
Because of \eqref{equ:appUnifomConv}, the right hand side converges to zero as $n\to\infty$ which proves \eqref{eqe:App_fStich_g}.
\end{IEEEproof}


\end{document}